\documentclass[english,cleveref,thm-restate,numberwithinsect]{lipics-v2021}

\usepackage{graphicx} % Required for inserting images
\usepackage{amsmath}
\usepackage{amssymb}
\usepackage{hyperref}
\usepackage{mathtools}
\usepackage{dsfont}
\usepackage{todonotes}
\usepackage{multirow}
\usepackage{gensymb}
\usepackage{booktabs}
\usepackage[dvipsnames]{xcolor}

\newcommand{\N}{\mathbb{N}}
\newcommand{\Z}{\mathbb{Z}}
\newcommand{\R}{\mathbb{R}}

\newcommand{\tOh}{\tilde{O}}

\DeclareMathOperator*{\argmin}{argmin}
\DeclareMathOperator*{\ind}{ind}
\DeclareMathOperator*{\cel}{cell}
\DeclareMathOperator*{\YES}{YES}

\DeclareMathOperator*{\RT}{RT}

\newcommand{\hddist}{\delta_{\overrightarrow{H}}}
\newcommand{\hddistun}{\delta_{H}}
\newcommand{\FOPZ}{\ensuremath{{\mathsf{FOP}_\mathbb{Z}}}}

\title{\texorpdfstring
  {Computing $L_\infty$ Hausdorff Distances Under Translations: 
   The Interplay of Dimensionality, Symmetry and Discreteness}
  {Computing L-infinity Hausdorff Distances Under Translations: 
   The Interplay of Dimensionality, Symmetry and Discreteness}}
\titlerunning{Computing Hausdorff Distances Under Translations} %TODO optional, please use if title is longer than one line

\author{Sebastian Angrick}{Karlsruhe Institute of Technology, Germany}{sebastian.angrick@kit.edu}{https://orcid.org/0009-0007-9840-9611}{}
\author{Kevin Buchin}{TU Dortmund, Germany}{kevin.buchin@tu-dortmund.de}{https://orcid.org/0000-0002-3022-7877}{}
\author{Geri Gokaj}{Karlsruhe Institute of Technology, Germany}{geri.gokaj@kit.edu}{https://orcid.org/0009-0002-7500-6848}{Partially supported by the Deutsche Forschungsgemeinschaft (DFG, German Research Foundation) - 462679611.}
\author{Marvin Künnemann}{Karlsruhe Institute of Technology, Germany}{marvin.kuennemann@kit.edu}{}{Partially supported by the Deutsche Forschungsgemeinschaft (DFG, German Research Foundation) - 462679611.}

\authorrunning{S. Angrick, K. Buchin, G. Gokaj, M. Künnemann} %TODO mandatory. First: Use abbreviated first/middle names. Second (only in severe cases): Use first author plus 'et al.'

\Copyright{Sebastian Angrick, Kevin Buchin, Geri Gokaj, Marvin Künnemann} %TODO mandatory, please use full first names. LIPIcs license is "CC-BY";  http://creativecommons.org/licenses/by/3.0/

 %TODO mandatory; please add comma-separated list of keywords

\funding{}%optional, to capture a funding statement, which applies to all authors. Please enter author specific funding statements as fifth argument of the \author macro.-

\acknowledgements{The authors would like to thank Jean Cardinal, Anita Dürr, Nick Fischer, Virginia Vassilevska Williams, and Karol Węgrzycki for stimulating discussions and contributions during the early stages of this work, as well as the Lorentz center and the organizers of its workshop on \emph{Fine-Grained \& Parameterized Computational Geometry} for facilitating these discussions.} %optional
\hideLIPIcs
\nolinenumbers

\EventEditors{John Q. Open and Joan R. Access}
\EventNoEds{2}
\EventLongTitle{42nd Conference on Very Important Topics (CVIT 2016)}
\EventShortTitle{CVIT 2016}
\EventAcronym{CVIT}
\EventYear{2016}
\EventDate{December 24--27, 2016}
\EventLocation{Little Whinging, United Kingdom}
\EventLogo{}
\SeriesVolume{42}
\ArticleNo{23}
\begin{document}

\maketitle

\begin{abstract}
    To measure the similarity of the \emph{shape} of point sets, rather than their mere closeness in space, various notions of a \emph{Hausdorff distance under translation} have been investigated. 
    Specifically, let $P$ and $Q$ denote point sets of $n$ and $m$ points, respectively, in $\mathbb{R}^d$. 
    We consider the task of computing the minimum distance $d(P,Q+\tau)$ over an admissible set of translations $\tau \in T$, where $d(\cdot, \cdot)$ denotes the Hausdorff distance under the $L_\infty$-norm. 
    As variants, we distinguish between \emph{continuous} ($T=\mathbb{R}^d$) or \emph{discrete} ($T$ is a given finite set of $t$ translations) as well as \emph{directed} or \emph{undirected} (choosing the directed or undirected Hausdorff distance for $d(\cdot, \cdot)$).

    We seek to apply the paradigm of fine-grained complexity to understand the complexity of these variants, and in particular: How is the running time influenced by the dimension $d$, the relationship between $n$ and $m$, and the specific choice of variant?
    As our main results, we obtain:
    \begin{itemize}
    \item The asymmetric definition of the most studied variant, the continuous directed Hausdorff distance, results in an \emph{intrinsically asymmetric} time complexity: While (Chan, SoCG'23) established a symmetric $\tOh((nm)^{d/2})$ upper bound for all $d\ge 3$ and proved it to be conditionally optimal for \emph{combinatorial} algorithms whenever $m \le n$, we show that this lower bound does not hold for the case $n \ll m$, by providing a combinatorial, almost-linear-time algorithm for $d=3$ and $n=m^{o(1)}$. We further prove \emph{general}, i.e., non-combinatorial, conditional lower bounds for $d\ge 3$, in particular: (1) $m^{\lfloor d/2 \rfloor -o(1)}$ for small $n$ and (2) $n^{ d/2 -o(1)}$ for $d=3$ and small $m$.
    \item We observe that the directed and undirected case is closely related, 
    in particular, all our lower bounds for $d \geq 3$ hold for both the directed and undirected variant. A remarkable exception is the case of $d=1$ for which we provide a conditional separation. Specifically, in contrast to the undirected variants being solvable in near-linear time~(Rote, IPL'91), we show that the directed variants are at least as hard as the additive problem MaxConv LowerBound introduced in (Cygan, Mucha, Wegrzycki and Wlodarczyk, TALG'19).
    \item We show that the discrete variants reduce to a variant of 3SUM for $d\le 3$. This gives a barrier in proving a tight lower bound of these variants under the Orthogonal Vectors Hypothesis (OVH); in contrast, the continuous variants admit a tight conditional lower bound under OVH in $d=2$ (Bringmann, Nusser, JoCG'21).
    \end{itemize}
    These results reveal an intricate interplay of dimensionality, symmetry and discreteness in determining the fine-grained complexity of computing Hausdorff distances under translation. 
\end{abstract}

\section{Introduction}

Consider the classic Hausdorff distance of given point sets $P$ and $Q$, which comes in two flavors: the \emph{directed} Hausdorff distance $\hddist(P,Q) \coloneqq \max_{p\in P}\min_{q\in Q} \|p-q\|$, as well as the \emph{undirected} Hausdorff distance $\hddistun(P,Q) \coloneqq \max\{\hddist(P,Q), \hddist(Q,P)\}$. Here and throughout the paper, we fix our metric space to be $(\mathbb{R}^d, L_\infty)$. The two variants are sometimes also referred to as one-sided/bidirectional or one-way/two-way. Furthermore, $\hddistun$ yields a metric, in contrast to $\hddist$. For any constant $d\in \mathbb{N}$ and point sets $P$ and $Q$ of $n$ and $m$ points in $\mathbb{R}^d$, we can compute these distance in near-linear time\footnote{Throughout the paper, we use $\tOh(\cdot)$, to hide polylogarithmic terms in $n$ and $m$. Since we assume $d$ to be constant, this includes factors of the form $\log^{O(d)} nm$.} $\tOh(n+m)$ using $L_\infty$-nearest neighbor search, e.g., in a range tree~\cite{kapoor1996new}.

A computationally more challenging task is to compute the Hausdorff distance under some transformation, i.e., to minimize the Hausdorff distance of $P'$ and $Q$ among all sets $P'$ obtained from $P$ using a set of admissible transformations. Such notions are well-studied from the perspective of \emph{shape matching} or \emph{geometric pattern matching}, see, e.g.~\cite{AltG00,AltMWW88}, where $P$ and $Q$ model geometric objects and we are interested in comparing their general shapes rather than their mere closeness in space. In this paper, we focus on arguably one of the most basic settings, i.e., computing $L_\infty$ Hausdorff distances under translation. Specifically, when $T$ denotes a set of admissible translations, the Hausdorff distance under translation is defined as $\min_{\tau \in T} \delta(P+\tau, Q)$ where $\delta$ denotes either the directed or undirected Hausdorff distance, i.e., $\delta \in \{\hddist, \hddistun\}$.

The most popular setting is $T=\mathbb{R}^d$, which yields the natural translation-invariant version of the Hausdorff distance; we refer to this as the \textbf{(continuous) Hausdorff distance under translation (HuT)}. This measure has received significant interest from the computational geometry community both for $L_\infty$ and other norms~\cite{rote1991ComputingMinimumHausdorff, HuttenlocherKS93, HuttenlocherKR93,Rucklidge96,ChewK98,chew1999GeometricPatternMatching,barequet_polygon_2011,BringmannN21,Chan23}, yet despite this interest, only for $d=2$  tight upper and conditional lower bounds are known (see below for details).

Another natural version is to let $T$ be a finite set given as part of the input. We refer to this as the \textbf{discrete Hausdorff distance under translation (DiscHuT)}. This version occurs naturally in the context of approximating the (continuous) Hausdorff distance under translation. More precisely, computing an $\alpha$-approximation for the continuous version can be achieved by choosing $f(\alpha,d)$ many translations $T$ such that 
$\min_{\tau \in T} \hddist(P+\tau, Q) \leq \alpha \min_{\tau\in \mathbb{R}^d} \hddist(P+\tau ,Q)$. Since we can compute the discrete Hausdorff distance under translation in time $\tOh(|T| (|P|+|Q|))$, this yields baseline $\alpha$-approximation algorithms for the continuous Hausdorff distance under translation with running time $\tilde{O}(f(\alpha,d) (n+m))$; see also~\cite{GoodrichMO99,wenk2003ShapeMatchingHigher}. Consequently, obtaining faster algorithms for the discrete Hausdorff distance under translation transfers to faster approximation algorithms for the continuous version. Furthermore, the discrete version recently received interest from the perspective of a structural fine-grained complexity theory: the directed discrete Hausdorff distance under translation is part of a problem pair that is fine-grained complete for the class \FOPZ, a class of polynomial-time problems formed from first-order properties on finite additive structures~\cite{gokaj_completeness_2025} (for further discussion, see Section~\ref{sec:apx-fopz-problems}).

Previous works on Hausdorff distances under translation mostly focus --with notable exceptions such as~\cite{BringmannN21,Chan23}-- on a single particular variant, often considering particular dimensions $d$ and the \emph{balanced} case, in which all inputs sets (i.e., $P$, $Q$, and in the discrete case also $T$) have roughly equal size $n$. In this work, we aim to give a more comprehensive analysis, which reveals an intricate interplay between the dimensionality, symmetry (directed/undirected and the relationship between $n$ and $m$), and discreteness.

\subsection{Our Results}

\newcommand{\newmark}{({\color{orange}$*$})~}
    
\begin{figure}[t]
    \centering
    \includegraphics[width=0.95\textwidth]{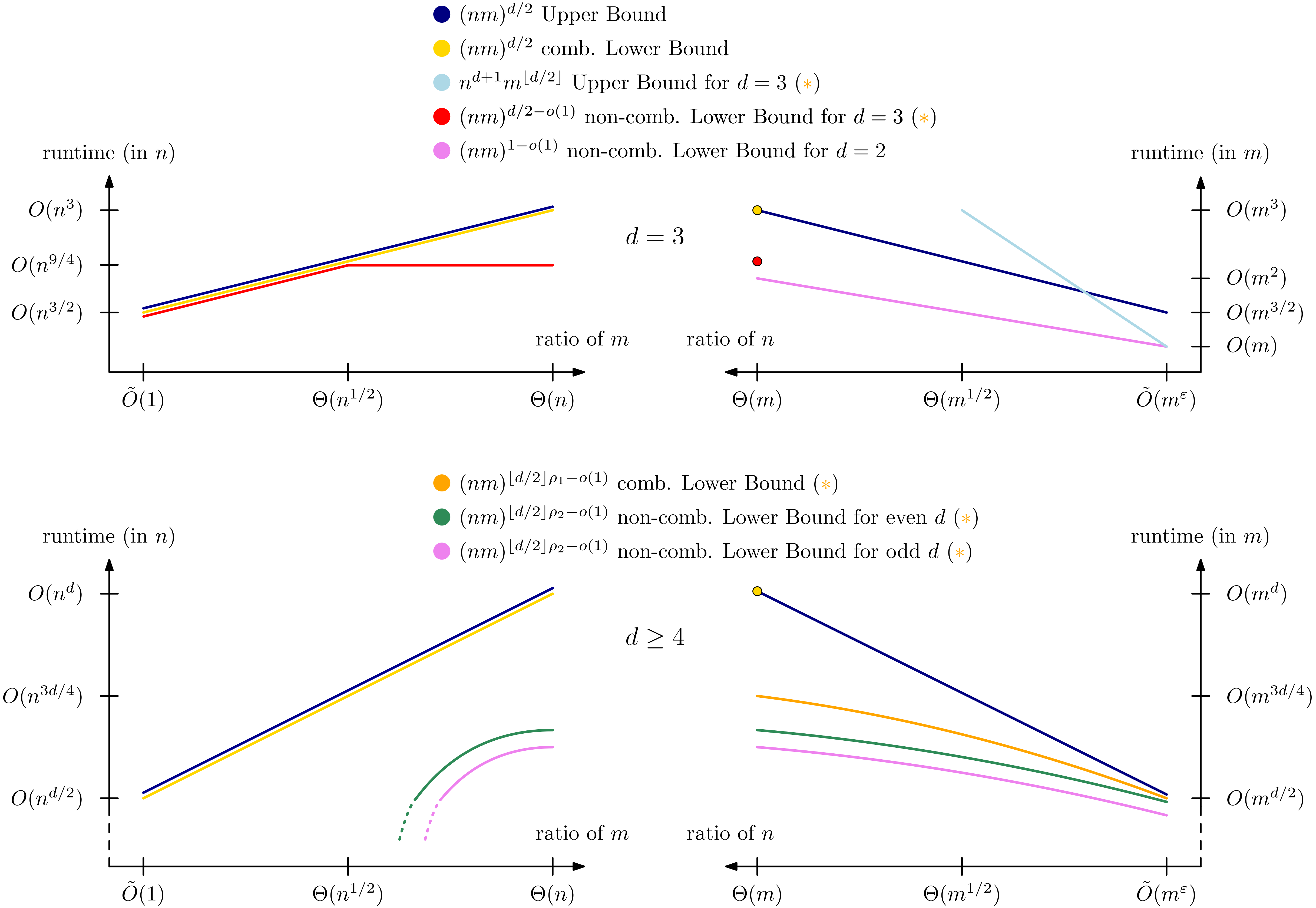}
    \caption{Schematic figure of upper and lower bounds for directed Hausdorff under Translation in $3$-D (upper figure) and $d$-D for $d \geq 4$ (lower figure), novel results are marked with \newmark. $\rho_1, \rho_2$ are functions on the ratio of $n,m$ which correspond to values described in \Cref{thm:lopsided-directed-LB}. }
    \label{fig:complexity}
\end{figure}

\begin{table}[t]
    \centering
    \newcommand{\ubColor}[1]{{\color{Blue} #1}}
    \newcommand{\lbLowConfColor}[1]{{\color{OliveGreen} #1$^\dagger$}}
    \newcommand{\lbHighConfColor}[1]{{\color{OliveGreen} #1}}

    %---- First matrix ----%
   
    \begin{tabular}{|c|c|c|}
        \toprule
        $d=1$ & \textbf{Directed} & \textbf{Undirected} \\
        \midrule
            \textbf{Continuous} & 
            \ubColor{$\tOh(n^2)$ (\cite{ChewK98})} ~
            \newmark \lbLowConfColor{$n^{2-o(1)}$ (Apx.~\ref{thm:necklace-LB})} &
            \ubColor{$\tOh(n)$ (\cite{rote1991ComputingMinimumHausdorff})} ~
            \lbHighConfColor{$\Omega(n)$}
            \\
        \midrule
        \textbf{Discrete} & 
        \ubColor{$\tOh(n^2)$ (Apx.~\ref{sec:baseline})} ~
        \newmark \lbLowConfColor{$n^{2-o(1)}$ (\ref{thm:maxconv-LB})} &
        \ubColor{$\tOh(n)$ (Apx.~\ref{sec:baseline})} ~
        \lbHighConfColor{$\Omega(n)$} 
        \\
        \bottomrule
    \end{tabular}
    
    \vspace{1em} % vertical space between the two matrices
    
    %---- Second matrix ----%
    \begin{tabular}{|c|c|c|}
        \toprule
        $d=2$ & \textbf{Directed} & \textbf{Undirected} \\
        \midrule
        \textbf{Continuous} & 
        \ubColor{$\tOh(n^2)$ (\cite{ChewK98})} ~
        \lbHighConfColor{$n^{2-o(1)}$ (\cite{BringmannN21})} &
        \ubColor{$\tOh(n^2)$ (\cite{ChewK98}, \ref{obs:undir-to-dir})} ~
        \lbHighConfColor{$n^{2-o(1)}$ (\cite{BringmannN21})}
        \\
        \midrule
        \textbf{Discrete} & 
        \ubColor{$\tOh(n^2)$ (Apx.~\ref{sec:baseline})} ~
        \newmark \lbLowConfColor{$n^{2-o(1)}$ (\ref{thm:reduction-orthants-to-undirected})} &
        \ubColor{$\tOh(n^2)$ (Apx.~\ref{sec:baseline})} ~
        \newmark \lbLowConfColor{$n^{2-o(1)}$ (\ref{thm:reduction-orthants-to-undirected})}
        \\
        \bottomrule
    \end{tabular}

    \caption{An overview of known results of the $L_\infty$-Hausdorff distance under Translation in $1$-D and $2$-D for the balanced case, i.e., $n = m$ ($=t$ in the discrete setting). 
    Upper Bounds are marked in blue, lower bounds in green, while results under non-standard hypothesis are additionally marked with ($\dagger$). Novel results are denoted with \newmark{}. 
    }
    \label{tab:results-d1-d2}
\end{table}

Our first question concerns the fine-grained complexity of the algorithmically most studied variant, the (continuous) directed Hausdorff distance under translation. The currently best upper bounds are $O(nm \log nm)$ for $d\le 2$~\cite{ChewK98}, $O((nm)^{d/2} \log nm)$ for $d=3$~\cite{chew1999GeometricPatternMatching}\footnote{While~\cite{chew1999GeometricPatternMatching} only considers the case $m=n$, their analysis also holds for distinct $n, m$. A short explanation of their approach can also be found in \Cref{sec:HuT-baseline-algorithms}.}, and $O((nm)^{d/2}(\log\log nm)^{O(1)})$ for general $d\ge 4$ \cite{Chan23}. This suggests the following question:
\begin{center}
\textbf{Question 1:} \emph{Is $(nm)^{d/2 \pm o(1)}$ the time complexity of directed Hausdorff $L_\infty$-distance under Translation?}    
\end{center}

A conditional lower bound of Bringmann and Nusser~\cite{BringmannN21} (based on the Orthogonal Vectors Hypothesis) answers this question positively for $d=2$. Subsequently, Chan~\cite{Chan23} gave evidence of optimality even for $d\ge 3$, by proving a conditional lower bound of $(nm)^{d/2-o(1)}$ for \emph{combinatorial} algorithms\footnote{Here, combinatorial algorithms refers to algorithms that avoid the use of algebraic techniques underlying fast matrix multiplication algorithms beating $n^{3-o(1)}$ running time.} in the case that $m\le n$. However, proving such a conditional lower bound against \emph{general} algorithms for $d\ge 3$ (or better conditional lower bounds than $(nm)^{1-o(1)}$ for the case $n\ll m$) remained open until this work. Furthermore, for $d=1$, no $\tOh(n+m+(nm)^{d/2}) = \tOh(n+m)$ algorithms are known, raising the question whether any superlinear lower bound can be proven for $d=1$. 

As our first main contribution, we give a perhaps surprising \emph{negative} answer to Question 1: For $d=3$ and $n=m^{o(1)}$, there exists an almost-linear time $m^{1+o(1)}$ algorithm, showing that in an unbalanced setting, the upper bound of $\tOh((nm)^{d/2})$ can indeed be beaten (\Cref{sec:apx-other-side}). 
Interestingly, this stands in contrast to the case of moderately small $m$: For $d=3$ and any $m\le \sqrt{n}$ we prove a tight conditional lower bound of $(nm)^{d/2-o(1)}$, by combining the proof techniques of Chan's combinatorial $k$-clique lower bound~\cite{Chan23} and the tight 3-uniform hyperclique lower bound for Klee's measure problem in 3D~\cite{kunnemann2022TightNoncombinatorialConditional} (\Cref{sec:apx-HuT-3D-new}). 
This reveals that the time complexity of the directed Hausdorff distance under translation is \emph{inherently asymmetric} in $n$ and $m$ and breaks with tradition for algorithmic results on the Hausdorff distance under translation. 

We investigate the case of small $n$ further and observe that our algorithmic approach could plausibly generalize to an $\tOh(n^{d+1} m^{\lfloor d/2 \rfloor})$ algorithm (we elaborate on the missing ingredient in \Cref{sec:apx-other-side}) -- such an algorithm would beat the baseline algorithm for all odd $d$ and small values of $n$.
We also show that such an algorithm would indeed be conditionally tight: we also prove a corresponding lower bound of $m^{\lfloor d/2 \rfloor - o(1)}$ for all $d\ge 4$, assuming the 3-uniform hyperclique hypothesis (\Cref{sec:apx-lopsided}).

Our conditional lower bounds can be adapted to the case of $n = \Theta(m)$, but no longer match the upper bound. In fact, we obtain a conditional lower bound of  $n^{2.25-o(1)}$ for $d=3$ which still leaves a gap of $n^{0.75\pm o(1)}$ to the upper bound of $\tOh(n^{3})$. It remains an open problem to settle the full time complexity for \emph{any} relationship between $n$ and $m$. An intriguing possibility consistent with our results is the existence of an $\tOh(n^{\lceil d/2 \rceil}m^{\lfloor d/2 \rfloor})$-time algorithm.

We turn to our second main question.
\begin{center}
\textbf{Question 2:} \emph{How does the choice of dimension, directed vs.\ undirected and discrete vs.\ continuous affect the time complexity?}
\end{center}

We observe that the time complexity of computing Hausdorff distances under translation exhibits an intricate interplay between the different aspects of the problem, illustrated in \Cref{tab:results-d1-d2} and \Cref{fig:complexity}.

\paragraph*{On directed vs undirected.}
A folklore reduction from computing the undirected Hausdorff distance to computing the directed Hausdorff distance under translation (\Cref{obs:undir-to-dir}) shows that the undirected case is never harder than the (balanced) directed case. Could there exist a reduction in the other direction?

We show that in full generality, such a reduction appears unlikely: While for all $d\ge 3$, we can transfer our conditional lower bounds for the directed case to analogous bounds in the undirected case (analogous to the OVH-based lower bound for $d=2$~\cite{BringmannN21}, which also transfers to the undirected case) (\Cref{sec:apx-undirected}), we obtain a conditional separation for $d=1$. Notably, Rote~\cite{rote1991ComputingMinimumHausdorff} obtains a $\tOh(n+m)$-time algorithm for computing the undirected Hausdorff distance under translation in $d=1$. 
For the directed setting, we give a fine-grained reduction from the quadratic-time $L_\infty$-Necklace Alignment problem~\cite{BremnerCDEHILPT14}, which in turn is at least as hard as the MaxConv LowerBound problem studied in~\cite{cygan2019ProblemsEquivalentMin+Convolution}. We remark that we are not aware of any conditional lower bound for MaxConv LowerBound under an established hypothesis from fine-grained complexity theory; however, the question of subquadratic equivalence between this problem and the well-studied $(\min,+)$-convolution problem was raised as an open problem in 2017~\cite{cygan2019ProblemsEquivalentMin+Convolution}.  We believe that our reduction to the (discrete or continuous) directed Hausdorff distance adds further reason to study the apparent quadratic-time complexity of MaxConv LowerBound, in hopes to either refute or justify this hardness barrier.

Circumventing the separation in $d=1$, we give a general reduction from the undirected setting to the directed setting that blows up the dimension (Section~\ref{sec:apx-undirected}). This allows us to transfer the lower bounds from the directed case for $d=1$ to the undirected case in $d=2$. In particular, we conclude that if MaxConv LowerBound truly requires quadratic-time, the undirected setting is near-linear time solvable if and only if $d=1$.

\paragraph*{On discrete vs.\ continuous.}
The tight lower bound based on OVH~\cite{BringmannN21} convincingly settles the fine-grained complexity of the continuous setting in $d=2$. In contrast, the only justification of hardness of the discrete setting is given by our reductions from MaxConv LowerBound and $L_\infty$-Necklace Alignment (see above and \Cref{sec:lower-bounds-1D}). Is there any reason why a tight lower bound based on OVH does not exist for the discrete setting?

We find evidence for the non-existence of such a reduction by showing a fine-grained reduction from (directed or undirected) discrete Hausdorff distance under translation in $d\le 3$ to the All-Ints 3SUM problem (see Section~\ref{sec:apx-fopz-problems} for details and definition). In particular, this implies a tight fine-grained reduction from the balanced setting to 3SUM. Thus, establishing quadratic-time OVH-hardness of the balanced, discrete setting implies a tight reduction from OV to 3SUM, which is a central open problem in fine-grained complexity theory, and suffers from barriers based on the Nondeterministic Strong Exponential Time Hypothesis, see~\cite{CarmosinoGIMPS16}.

\medskip

We remark that we obtain further connections from the class of additive problem $\FOPZ$ to the discrete variants. For space reasons, these results are deferred to Section~\ref{sec:apx-additive-problems}.

%%%%%%%%%%%%%%%%%%%%%%%%%%%%%%%%%%%%%%%%%%%%%%%%%%%

\section{Preliminaries}
\label{sec:preliminaries}
We include a brief definition of relevant problems as well as the used hardness hypotheses in \Cref{sec:definitions}.
We further give a brief overview of known algorithms for the continuous and discrete variant in \Cref{sec:baseline}.

We recall the Hausdorff distance as a distance measure between two point sets $P, Q$ of sizes $n, m$.
The metric space used for this paper will be $(\R^d, L_\infty)$. 
The undirected (symmetric) and directed (asymmetric) variant are defined as:
\begin{align*}
    \hddist(P,Q) &\coloneqq \max_{p\in P}\min_{q\in Q} \|p-q\|_\infty & \hddistun(P,Q) &\coloneqq \max\{\hddist(P,Q), \hddist(Q,P)\}
\end{align*}

We show almost all of our conditional lower bounds already for the \emph{decision variant} of Hausdorff under Translation, which receives a test distance $\delta > 0$ as additional input\footnote{Note that the case $\delta = 0$, i.e., finding an exact match, allows an almost-linear-time algorithm \cite{fischer2024BaurStrassen} and is thus uninteresting for our purposes.}. Clearly, the corresponding optimization problems, finding a translation minimizing $\delta$, are at least as hard. Specifically, we study:
\begin{definition}[Hausdorff under Translation (\emph{HuT})] Given sets $P,Q\subset \mathbb{R}^d$ of $n$ and $m$ points, respectively, and a real value $\delta > 0$. Decide whether: \[ 
        \exists{\tau \in \mathbb{R}^d} : \hddist(P+\tau, Q) \leq \delta \quad\equiv \quad \exists{\tau \in \mathbb{R}^d} ~ \forall {p\in P} ~ \exists{q\in Q} : \| (p +\tau) - q\|_\infty \leq \delta
    \]
\end{definition}
\begin{definition}[Discrete Hausdorff under Translation (\emph{DiscHuT})] Given sets $P,Q\subset \mathbb{Z}^d$ of $n$ and $m$ points, respectively, a real value $\delta > 0$ as well as a set of at most $t$ translations $T\subset \mathbb{Z}^d$. Decide whether:
    \[ 
        \exists{\tau \in T} : \hddist(P+\tau, Q) \leq \delta \quad \equiv \quad \exists{\tau \in T}~\forall{p\in P} ~ \exists{q\in Q}: \| (p+\tau) - q\|_\infty \leq \delta
    \]
\end{definition}
Note that our definition of the discrete variant assumes all input entries to be polynomially bounded integers\footnote{One may also assume rational coordinates, which is equivalent by a scaling argument.}.

To obtain the \emph{undirected} versions of the above problems, we replace $\hddist(P + \tau, Q)$ by $\hddistun(P + \tau, Q)$, and denote the problem by \emph{(Disc)uHuT}. If not specified otherwise, we consider the directed versions.
We call the case in which all input sets have roughly equal size, i.e., $Q = \Theta(|P|)$ and additionally for the discrete case $|T|=\Theta(|P|)$, the \emph{balanced} parameter setting, otherwise we say the parameters are \emph{lopsided}.

Lastly, we give a folklore reduction from undirected to directed Hausdorff under Translation that works for both the continuous and discrete variant. 
By setting $\sigma=(M, 0,\dots, 0)$ for a sufficiently large $M$, we see that $\hddistun(P + \tau, Q) = \hddist((P \cup (\sigma - Q)) + \tau, Q \cup (\sigma-P))$.
\begin{observation}
    \label{obs:undir-to-dir}
    The undirected (Discrete) Hausdorff under Translation problem with input sizes $n,m$ reduces to the directed version with input sizes $n' = m' = n+m$.
\end{observation}
While this reduction introduces additional points to the input, we show in \Cref{sec:apx-undirected} that an input size maintaining reduction to a single instance of directed HuT is unlikely.

As our model of computation, we use the Real RAM for continuous HuT, while for the discrete HuT variant with point sets in $\Z^d$ we assume a Word RAM machine model with word size $\Theta(\log n)$.

\subsection{Translation Problems}
\label{sec:translation-problems}
For clarity of presentation, it is often useful to view the Hausdorff under Translation decision problem as a translation problem involving certain \emph{shapes}.
\begin{definition}[Translation Problem] 
    Given a set of $n$ points $P$ and $m$ shapes $\mathcal Q$ in $\R^d$. Compute a translation $\tau \in T \subseteq \R^d$ such that \[
        \forall p \in P ~ \exists Q \in \mathcal Q: p + \tau \in Q.
    \]
\end{definition}
For an instance of HuT with point sets $P, Q$ and $\delta > 0$, if we construct hypercubes $\mathcal Q$ with side length $2\delta$ and centers $q \in Q$, the region $T^* = \bigcap_{p \in P} \bigcup_{\hat Q \in \mathcal Q} \hat Q - p$ is exactly the region of feasible translations, i.e., we have $\tau \in T^* \iff \forall p\in P ~ \exists q \in Q: \hddist(P + \tau, Q) \leq \delta$ \cite{agarwal_applications_1994}.
We call this case, where all shapes $\mathcal Q$ are hypercubes of the same size, the \emph{Translation Problem with Hypercubes (TPwC)}. In case all shapes of $Q$ are axis-aligned boxes (of arbitrary size), we call it the \emph{Translation Problem with Boxes (TPwB)}.
While the TPwC problem is equivalent to HuT, there is no reduction to HuT known within the same ambient space from the Translation Problem with Boxes or even from the translation problem with hypercubes of different sizes.
We further use a third variant. 

\begin{definition}[Translation Problem with Orthants (TPwO)]
    Given a target box $B$ with a maximum side length $\delta_0$.
    Given a set of $k$ Translation Problem with Hypercube instances, each on a point set $P_i$ and hypercubes $\mathcal Q_i$ with side length $\delta > \delta_0$, such that $\sum_{i \in [k]}|P_i| = n$, $\sum_{i \in [k]}|\mathcal Q_i| = m$.
    Compute a translation $\tau \in \R^d$ such that $\forall i \in [k]: P_i + \tau \subset B$ and \[
    \forall i\in [k] ~ \forall p_i \in P_i ~\exists Q_i \in \mathcal Q_i: p_i + \tau \in Q_i.
    \]
\end{definition}

We call the shapes of the TPwO instance, i.e., the shapes of its sub-instances, orthants as within the target box all hypercubes are indistinguishable from orthants.
An important property of this problem, which we will use later, is that we may fix the side-length of all orthants to a value $\delta > \delta_0$.
Similar to \cite[Lemma 3]{Chan23}, the TPwO problem reduces to the TPwC and thus has a reduction to HuT.
A converse reduction, from HuT to TPwO, is generally not known.

\begin{lemma}
    \label{lem:tpwo-to-tpwc}
    The Translation Problem with Orthants on $n$ points and $m$ orthants reduces in linear time to the Translation Problem with Hypercubes on $\Theta(n)$ points and $\Theta(m)$ hypercubes.
\end{lemma}

\begin{proof} 
    For the TPwO instance, a set of TPwC sub-instances $(P_1, \mathcal Q_1), ..., (P_k, \mathcal Q_k)$ and a target box $B$ with maximum side length $\delta_0$ is given. 
    We need to show how to encode the target box as well as the separate sub-instances into a single TPwC instance.
    
    For the target box, restricting the translated point set, we note that it implies a bounding box of the available translations: For each dimension, take the corresponding boundaries of the target box and its nearest points, their difference is the bounding box for the available translations.
    Having the bounding box for translations, we introduce two more TPwC sub-instances.
    By definition of TPwO, we assume that the side length of the hypercubes in the TPwC instances is $\delta$ with $\delta > \delta_0$, so it is larger than the side length of the bounding box for the available translations.
    The newly introduced TPwC sub-instances each get one point and one hypercube, so that the only feasible translations that map both points inside their hypercubes are exactly the bounding box of feasible translations.
    
    For the encoding of the separate TPwC sub-instances with side length of $\delta$, we assume that all hypercubes intersect the target box, else we remove these hypercubes.
    We translate each sub-instance $(P_j, \mathcal Q_j)$ by a vector $(j9\delta, 0, ... 0)$, such that no target box and hypercube of different sub-instances overlap. 
    Thus under all feasible translations, a point in $P_j$ must lie in a box of $\mathcal Q_j$.
    Else, there exists a point set, either $P_1$ or $P_k$, for which the translated points do not lie in any hypercube of $\mathcal Q_1 \cup ... \cup \mathcal Q_k$.

    Any translation is now equivalently feasible to the original TPwO instance as for the constructed TPwC instance.  
\end{proof}

%%%%%%%%%%%%%%%%%%%%%%%

% Lower and upper bounds
\section{Complexity of Directed HuT in Higher Dimensions \texorpdfstring{$d \geq 3$}{d >= 3}}
\label{sec:apx-higher-dimension-overall}
We turn towards exploring the complexity landscape of (continuous) directed Hausdorff under Translation in higher dimension $d \geq 3$. 
In particular, we prove novel non-combinatorial lower bounds and devise a faster algorithm for a lopsided setting. Both lower bounds and the algorithm are tight for specific parameter regimes.

\begin{figure}
    \centering
    \includegraphics[width=0.9\linewidth]{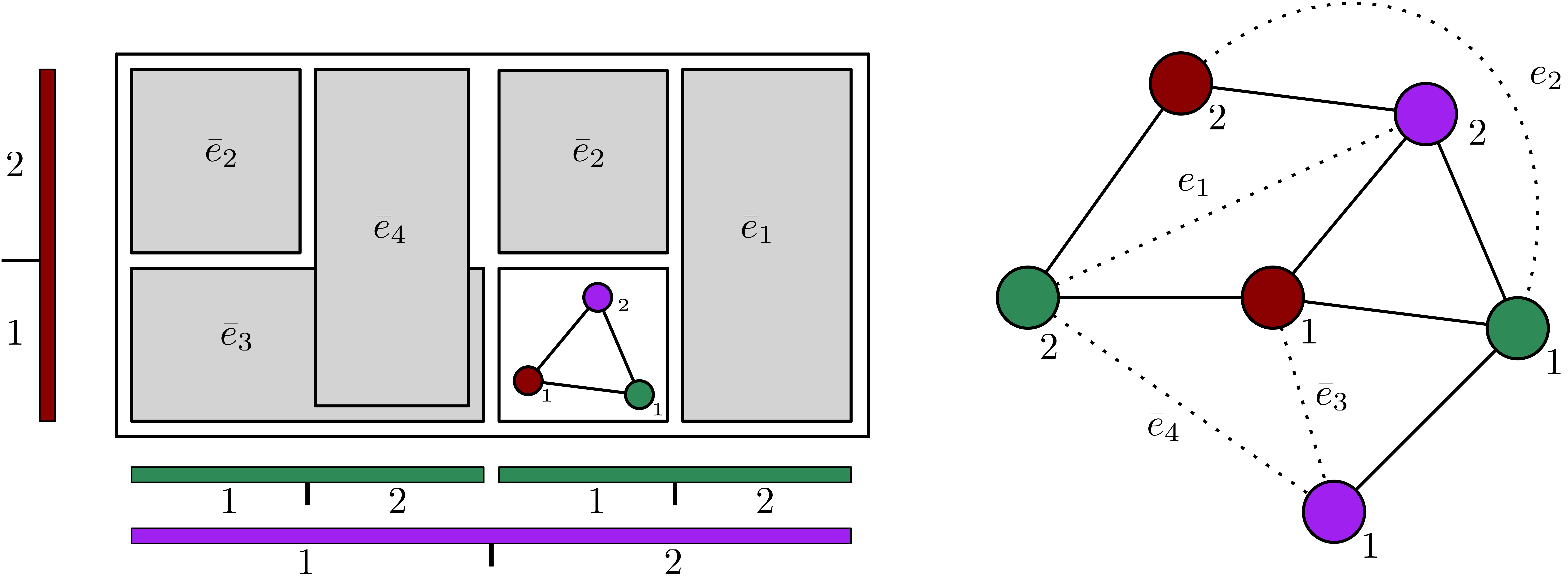}
    \caption{An example of an encoding of a $3$-partite graph to a $2$-D box (skewed cube). Normal lines represent edges, dotted lines represent a subset of non-edges. The gray boxes correspond to the forbidden regions of the respective non-edge. Note that in our reductions, we aim to cover the complement of these regions, the feasible regions.}
    \label{fig:basic-reduction-idea-app}
\end{figure}

For our lower bounds, we use a similar approach to Chan's reduction from $k$-clique to Hausdorff distance under translation~\cite{Chan23}, 
which encodes any $k$-tuple of vertices of a graph $G = (V,E)$ as cells inside a $d$-dimensional hypercube, see \Cref{fig:basic-reduction-idea-app}. 
For any non-edge $\bar e \in {V \choose 2} \setminus E$, he covers exactly the cells within the hypercube for which the corresponding $k$-tuple does not include $\bar e$, we call these cells the \emph{feasible region} of $\bar e$. 
In contrast, the cells for which the corresponding $k$-tuple includes all vertices of $\bar e$ are called the \emph{forbidden region}.
As a $k$-clique is a $k$-tuple which does not contain any non-edge, a cell that lies within the feasible regions of all $\bar m$ non-edges, and is thus covered $\bar m$ times, encodes a $k$-clique.

We use the same approach but reduce from the $3$-uniform $k$-hyperclique problem -- the corresponding hypothesis is that no (non-combinatorial) $O(n^{k - \varepsilon})$-time algorithm exists. For a formal definition of the problem and its hardness hypothesis, see~\Cref{sec:definitions}. 
Accordingly, for each of the $\bar m$ non-edges $\bar e$ of our hypergraph\footnote{Even though we consider hypergraphs with uniformity $u > 2$, we will refer to hyperedges as edges.} we cover the feasible region of $\bar e$, i.e., the cells of our hypercube which do not include $\bar e$.
We give a reduction from finding a cell inside all feasible regions to directed Hausdorff under Translation with appropriate parameters.

We start in \Cref{sec:apx-lopsided} with new combinatorial and non-combinatorial lower bounds that apply for almost arbitrary distributions of $n, m$. Note that all previous lower bounds for HuT in $d$-D with $d \geq 3$ assume $n \geq m$, to the best of our knowledge.
We further consider directed HuT with more balanced parameters $n, m$ in $3$-D. 
In particular, we prove a new non-combinatorial lower bound for $m = O(n^{1/2})$ in \Cref{sec:apx-HuT-3D-new} that matches the previous best, tight, but combinatorial lower bound of Chan \cite{Chan23}.

In \Cref{sec:apx-other-side}, we present a novel algorithm which beats the currently best $O((nm)^{d/2})$ algorithm \cite{Chan23} in the case of $m \gg n$. This runtime also matches our lower bounds for $n = \Theta(m^\varepsilon)$ for a sufficiently small but constant $\varepsilon > 0$.

\subsection{(Non-Combinatorial) Lower Bounds for Arbitrary \texorpdfstring{$n, m$}{n, m}}
\label{sec:apx-lopsided}

In this section, we prove the first of our lower bounds for directed HuT in $d \geq 4$ dimensions, which applies for nearly arbitrary distributions of $n,m$. The reduction has three steps. 

In a first step, given an instance of the $k$-hyperclique problem in $3$-uniform hypergraphs, we encode it into a $\lfloor d/2 \rfloor$-dimensional hypercube. The used encoding was introduced by Chan \cite{Chan23}.
For each non-edge of our input hypergraph, we construct a covering of its feasible region with axis-aligned boxes\footnote{As we only consider axis-aligned boxes in this paper, we may only call these boxes.}, see \Cref{lem:apx-feasible-region-boxes} for the exact encoding and covering. 
We shall prove that our hyperclique instance is $\YES$ if and only if there is a cell inside our hypercube that lies within the feasible region of all non-edges. 
In \Cref{lem:feasible-region-decomposition} we reduce the problem of finding a cell located within all feasible regions to an instance of the Translated (Box-)Shape Problem, see \Cref{def:translated-box-shape}, with $n$ translates of common shapes consisting of $m$ boxes, for our desired values of $n, m$. 

In the second step, the resulting Translated Shape Problem is reduced to the Translation Problem with Boxes, see \Cref{lem:shape-translation-reduction}.

Lastly, in our third step we reduce the Translation Problem with Boxes in $d$-D to the Translation Problem with Orthants in $2d$ dimensions with $n$ points and $m$ orthants, see \Cref{lem:box-to-cube}.
As given in \Cref{lem:tpwo-to-tpwc}, the Translation Problem with Orthants has a straightforward reduction to directed Hausdorff under Translation, which yields our lower bound.  

\paragraph*{Step 3: From Boxes to Orthants}
We start with the last step of our lower bound proof, which is reducing the Translation Problem with Boxes to the Translation Problem with Orthants. 
We implement the natural idea to cover boxes of arbitrary aspect ratios by orthants (large hypercubes), at the cost of doubling the underlying dimension. 
As this basic idea allows additional, unintended solutions, we need a technical translation gadget that ensures soundness. 

The basic idea of our translation gadget is to scale and rotate each dimension $i$ by $45\degree$, introducing a second dimension each. We denote the resulting dimensions by $i_1, i_2$. 
Any point is scaled and rotated accordingly, so all points of $\R$ in dimension $i$ lie on the diagonal $\{(x,x)\mid x \in \R\}$ in dimensions $i_1, i_2$.
Each box, which in dimension $i$ is represented as an interval $[a, b]$, is split into two intervals $[a, \infty)$ in dimension $i_1$ and $(-\infty, b]$ in dimension $i_2$, which describes an orthant. 
The proof idea is that a translation that for all dimensions $i_1, i_2$ stays on the diagonal, is equivalently a solution to the Translation Problem with Boxes as well as to the Translation Problem with Orthants.
To make this transformation usable for our purposes, we need $2^d$ (mirrored) copies of the constructed orthants as possible translations do not need to constrict themselves to the desired diagonal.
\Cref{fig:box-to-orthant-app} gives a visual intuition of the proof in $1$-D.

\begin{figure}
    \centering
    \includegraphics[width=0.9\linewidth]{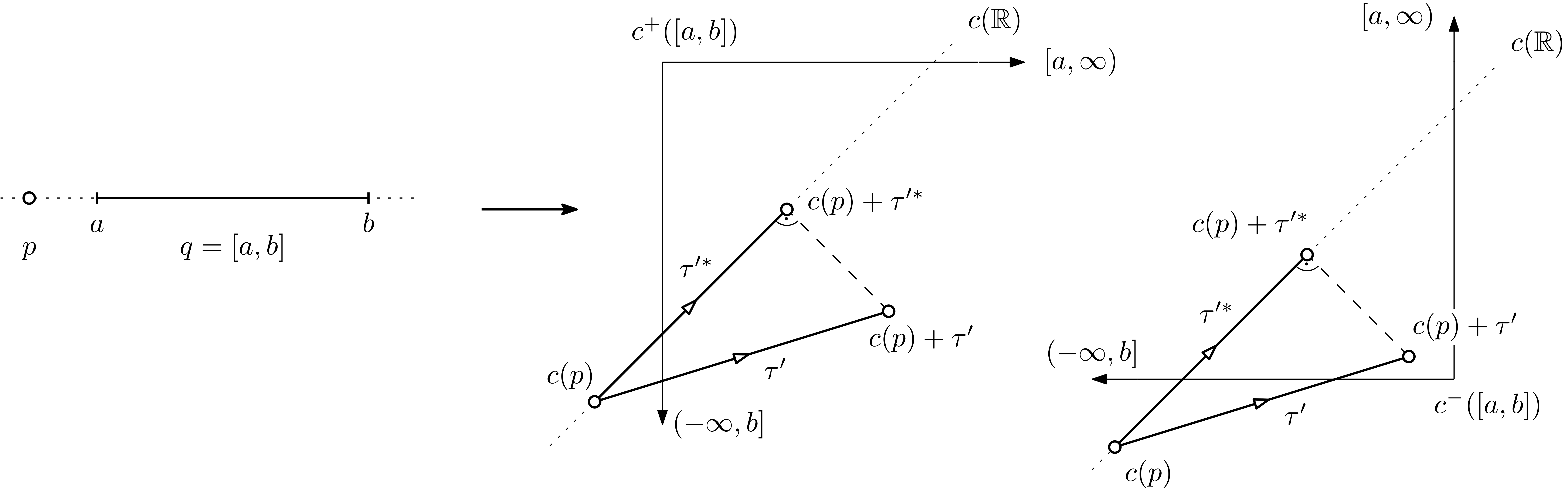}
    \caption{The reduction of a box in $1$-D, i.e., an interval $[a,b]$, to orthants in $2$-D. 
    We prove that for each feasible translation $\tau'$ there exists a feasible translation $\tau'^*$ which lies on the diagonal $c(\R)$. For that, we always assume that the translation $\tau '$ lies between the diagonal $c(\R)$ and the vertex of the orthant, i.e., for this example the orthant of $c^-$ is relevant. 
    }
    \label{fig:box-to-orthant-app}
\end{figure}

\begin{lemma}
    \label{lem:box-to-cube}
    The Translation Problem with Boxes in $d$-D with input parameters $n,m$ reduces in linear time to the Translation Problem with Orthants in $2d$-D with input parameters $n' = \Theta(n), m' = \Theta(m)$.
\end{lemma}

\begin{proof}
    \label{prf:box-to-cube}
    We define the coordinate-wise transformations 
    \begin{align*}
        c : p \mapsto (p, p), && c^+: [a,b] \mapsto [a, \infty) \times (-\infty, b], && c^-: [a,b] \mapsto (-\infty, b] \times [a, \infty). 
    \end{align*}
    For points in $\R^d$, we define the transformation $f$ as the coordinate-wise application of $c$. 
    For our analysis, we further apply $f$ to sets or intervals in which case we apply $f$ element-wise. 
    For boxes in $\R^d$, we define the transformations $f_1, ..., f_{2^d}$ as all possible combinations of the coordinate-wise application of $c^+, c^-$.
    Note that the images of $f_i$ are orthants.

    Given an instance of the Translation Problem with Boxes with $P$ being the point set and $\mathcal{B}$ being the set of boxes. 
    We construct the sub-instances $(f(P), f_1(\mathcal B)), ..., (f(P), f_{2^d}(\mathcal B))$ of the resulting Translation Problem with Orthants instance, the target box is simply the bounding box of $\mathcal B$. We claim that \[ 
        \exists {\tau \in \R^d} ~ \forall p \in P ~ \exists B \in \mathcal B: p + \tau \in B 
        \iff \exists {\tau' \in \R^{2d}}~\forall i\in [2^d] ~ \forall p' \in f(P) ~\exists C \in f_i(\mathcal B) : p' + \tau' \in C. 
    \]
    
    First, we see that if a translation for the right side is restricted to $f(\R^d)$, we get the desired equivalence irrespective of which $i \in [2^d]$, i.e., function $f_1,...,f_{2^d}$, we use. Formally, we have 
    \begin{equation}
        \label{eq:box-orthants}
        \forall p\in P, ~ B \in \mathcal B, ~ i \in [2^d]: \quad p + \tau^* \in B \iff f(p) + f(\tau^*) \in f_i(B)
    \end{equation}
    because $f(p) + f(\tau^*) \in f(\R^d)$ and $f_i(B) \cap{f(\R^d)} = f(B)$ for any $i \in [2^d]$, i.e., $f_i(B)$ covers in each dimension the same interval on $f(\R)$ as $B$ on $\R$.

    Let $\tau^*$ be a translation that fulfills the left side of the equivalence. We directly get that $f(\tau^*)$ is a feasible translation for the right side.

    For the other direction, we prove that if there is a fulfilling translation $\tau' \in \R^{2d}$ for the right side, there also is a fulfilling translation $\tau'^* \in f(\R^d)$.
    Assuming that the fulfilling translation $\tau' = f(\tau^*)$ for some $\tau^* \in \R^d$, by \Cref{eq:box-orthants} we get that $\tau^*$ fulfills the left side.
   
    Let $\tau'^* = f(\tau^*)$ be the point closest to $\tau'$ on $f(\R^d)$.
    We choose $i \in [2^d]$ such that $f_i$ in dimension $\ell \in [d]$ applies $c^+$ if $\tau'_{2\ell-1} \leq \tau'_{2\ell}$, and $c^-$ otherwise, i.e., we assume that $\tau'$ always lies on the same side of $f(\R)$ where the vertices of the orthants $f_i(\mathcal B)$ are. 
    Assume that in dimension $\ell$ we have $\tau'_{2\ell-1} \leq \tau'_{2\ell}$, otherwise switch the roles of $2\ell$ and $2\ell -1$ as well as $c^+, c^-$ in the following.
    
    For a point $p'=f(p) \in P'$ let $C = f_i(B)$ such that $p' + \tau' \in C$.
    Assume that $C$ has the vertex $(a', b')$ in dimensions $2\ell-1, 2\ell$.
    Since $C=f_i(B)$ covers $p' + \tau'$ and we have assumed that $f_i$ applies $c^+$ in dimension $\ell$, we know that $a' \leq (p' + \tau')_{2\ell -1}$ and $(p' + \tau')_{2\ell} \leq b'$.
    As we have chosen $\tau'^*$ to be the closest point to $\tau'$ on $f(\R^d)$ and in dimension $\ell$ we know $\tau'_{2\ell-1} \leq \tau'_{2\ell}$, we get that \[
        \tau'_{2\ell-1} \leq \tau'^*_{2\ell-1} = \tau'^*_{2\ell} \leq \tau'_{2\ell}
    \] 
    and can thus combine 
    \begin{equation}
        \label{eq:box-orthants-2}
        a' \leq (p' + \tau')_{2\ell-1} \leq (p' + \tau'^*)_{2\ell-1} = (p' + \tau'^*)_{2\ell} \leq (p' + \tau')_{2\ell} \leq b',
    \end{equation}
    proving that when projected to dimensions $2\ell-1, 2\ell$, $p'$ is covered by $C$ when translated by $\tau'^*$ instead of $\tau'$.
    As we have chosen $i$ such that \Cref{eq:box-orthants-2} holds for all dimensions $2\ell-1, 2\ell$ for $\ell \in [d]$, we know $p' + \tau'^* \in C$.
    By \Cref{eq:box-orthants} and $p' = f(p), \tau'^* = f(\tau^*), C = f_i(B)$ as defined, we get that $p + \tau^* \in B$,
    thus $\tau^*$ is a fulfilling translation to the left side of our equivalence.

    Note that the constructed instance $P', \mathcal C$ has size $|P'| = 2^d|P|$ and $|\mathcal{C}| = 2^d |\mathcal B|$. Since $2^d = \Theta(1)$, the result follows.
\end{proof}

\paragraph*{Step 2: The intermediate Translated Box-Shape Problem}
We introduce an intermediate problem, the Translated Shape Problem, similar to the one of \cite[Problem 4]{Chan23}. In contrast to \cite{Chan23}, it does not operate with hypercubes but arbitrary boxes. Still, the reduction from the Translated Shape Problem to the Translation Problem with Boxes is analog to \cite[Lemma 3]{Chan23}.

\begin{definition}
    \label{def:translated-box-shape}
    \textsf{\emph{(Translated (Box-)Shape Problem).}} 
    Let $\mathcal Z$ be a set of shapes in $\mathbb{R}^d$, composed of a union of boxes. Let $m$ be the total number of boxes over all shapes of $\mathcal Z$. 
    Given a set $\mathcal S$ of $n$ objects where each object is the translate of some shape in $\mathcal Z$, decide whether $\bigcap_{S \in \mathcal S} S \neq \emptyset$. 
\end{definition}

\begin{lemma}
    \label{lem:shape-translation-reduction}
    The Translated (Box-)Shape Problem with $n$ translation objects and a total of $m$ boxes reduces to the Translation Problem with Boxes with $|P| = \Theta(n)$ and $|\mathcal Q| = \Theta(m)$.
\end{lemma}

\begin{proof}
\label{prf:shape-translation-reduction}
    Let $\mathcal Z$ be the set of shapes and let the translations objects of $\mathcal S$ consist of a translation vector and a pointer to the shape.
    We rescale our shapes and translation objects so that they completely lie within the $[0, 1/2]^d$ hypercube and thus any point in $\bigcap_{S \in \mathcal S}S$ lies within $[0, 1]^d$. 

    We construct an instance of the Translation Problem with Boxes with inputs $P, \mathcal Q$.
    Let $Z_1, ..., Z_\ell$ be the shapes of $\mathcal{Z}$ that are used by at least one translation object and let $\mathcal B(Z_i)$ be the boxes of a shape $Z_i$.
    Let $u_i = (10i, 0, ..., 0) \in \R^d$ be a translation vector spatially separating the shapes of $\mathcal{Z}$.
    For each $B \in \mathcal B(Z_i)$, we add the box $B + u_i$ to $\mathcal Q$.
    For each $(t, Z_i) = S \in \mathcal S$, we add the point $u_i-t$ to $P$.

    For correctness, observe the following relation between a feasible solution of the Translated Shape Problem to the transformed Translation Problem on $P, \mathcal Q$.  
    \begin{align*}
        \tau \in\bigcap_{S \in \mathcal{S}}S
        \iff& \forall S=(t, Z_i) \in \mathcal S ~ \exists B \in\mathcal B(Z_i): \tau \in B + t \\
        \iff &\forall S=(t, Z_i) \in \mathcal S ~ \exists B \in \mathcal B(Z_i): u_i - t + \tau \in B + u_i \tag{*}\\
        \iff & \forall p \in P ~ \exists B' \in \mathcal Q: p + \tau \in B'.
    \end{align*}
    For the reverse direction assuming a feasible solution $\tau$, note that $u_i - t = p \in P$ and $B + u_j = B' \in \mathcal Q$ with $p + \tau \in B'$ implies $u_i = u_j$ as otherwise there is an uncovered point.
    For the sake of contradiction assume $u_j > u_i$, for $u_j < u_i$ the same argumentation applies with $(t, Z_1)$.
    First, $u_j > u_i$ implies $(\tau)_1 > 1$ as $B, t \subseteq [0, 1/2]^d$. 
    For the same reason, a point $p \in P$ corresponding to a translation object $(t, Z_\ell)$ would not be mapped inside any box because $(p + \tau)_1 = 10\ell - (t)_1 + (\tau)_1 > 10\ell+1/2$ but $\bigcup_{B' \in \mathcal Q} (B)_1 \subseteq [10, 10\ell + 1/2]$.
    
    As the remaining steps follow from definitions, we proved the equivalence statements.
\end{proof}

Lastly, we give a simple observation allowing, by repeated application, to compose arbitrary many individual Translated (Box-)Shape Problems without additional overhead. 

\begin{observation}
    \label{lem:shape-translation-composition}
    Given two instances of the Translated Shape Problem with inputs $\mathcal{S}_1, \mathcal{Z}_1$ and $\mathcal{S}_2, \mathcal{Z}_2$ of size $n_1, m_1$ and $n_2, m_2$ in a shared universe. 
    Deciding whether $\bigcap_{S \in \mathcal S_1 \cup \mathcal S_2}S \neq \emptyset$ reduces to a single instance of the Translated Shape Problem with $\mathcal{S} = \mathcal{S}_1 \cup \mathcal S_2$, $\mathcal Z = \mathcal{Z}_1 \cup \mathcal Z_2$ of size $n \leq n_1 + n_2$ and $m \leq m_1 + m_2$.
\end{observation}

\paragraph*{Step 1: From $k$-Hyperclique to the Translated Box-Shape Problem}
We now turn to the remaining reduction, which follows the basic idea presented in \Cref{sec:apx-higher-dimension-overall}.
Given a hypergraph on $n$ vertices, recall that the feasible region of a non-edge $\bar e$ are the cells within the hypercube that correspond to $k$-tuples which do not include $\bar e$. 
Note that we defined the complement, the forbidden region of $\bar e$, as the cells whose $k$-tuple includes $\bar e$.
A cell that lies in the feasible regions of all non-edges encodes a $k$-hyperclique.
In the course of the next three lemmas, we give a reduction from finding a cell that lies in the feasible regions of all non-edges to the Translated (Box-)Shape Problem with a desired distribution of parameters $n, m$.
As our approach will also provide novel combinatorial lower bounds, we may more generally consider a $u$-uniform hypergraph with $u \in \{2,3\}$.

We first give the encoding of the hypergraph to a $d$-dimensional hypercube and show that the feasible regions of the non-edges can be covered efficiently. 
The same encoding for a graph can exemplarily be seen in \Cref{fig:basic-reduction-idea-app}, note that for illustrative reasons we there have not covered the feasible but the forbidden regions.

\begin{figure}
    \centering
    \includegraphics[width=0.8\linewidth]{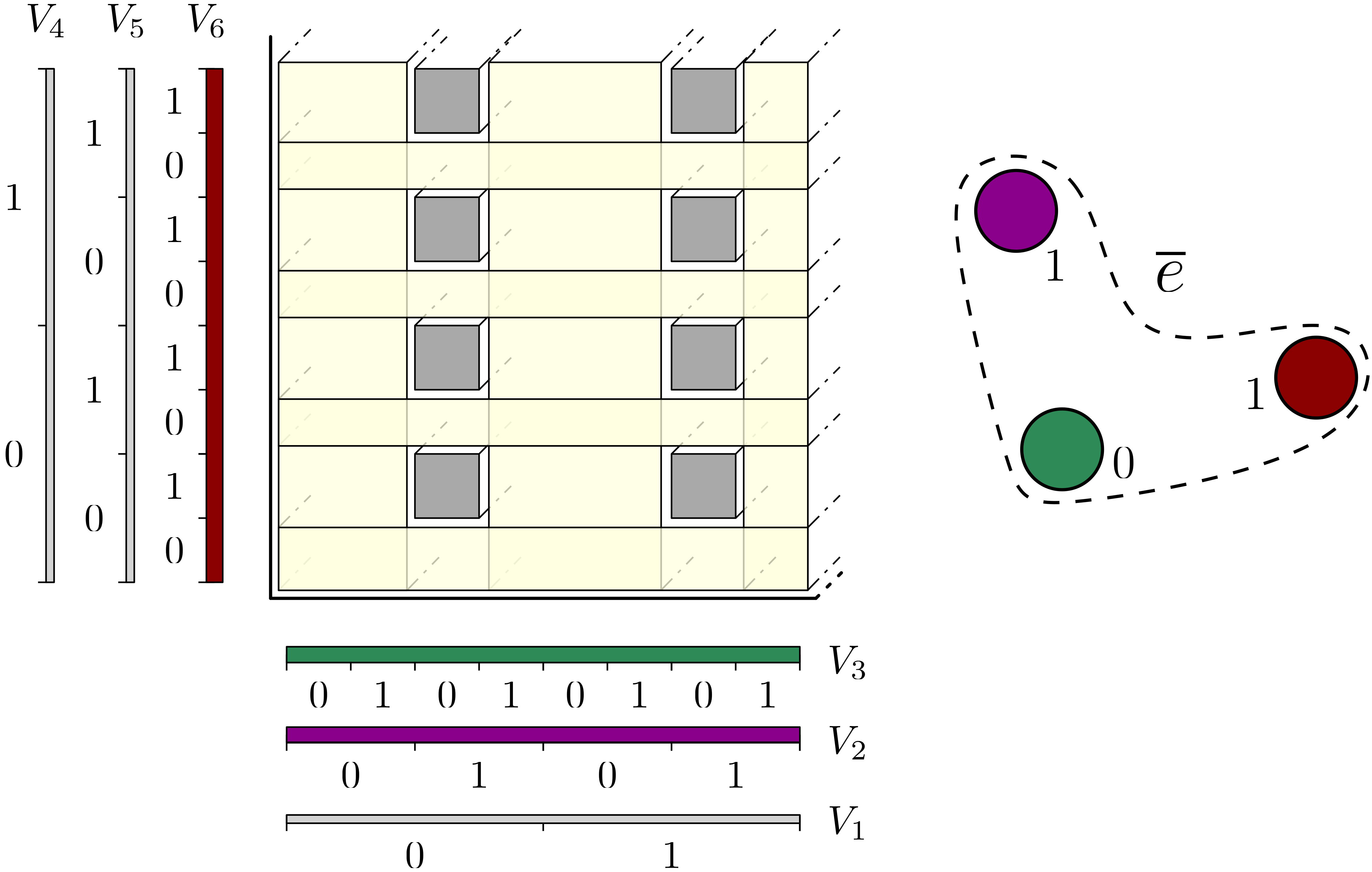}
    \caption{An example of the feasible region of a non-edge $\bar e \in V_2 \times V_3 \times V_6 \setminus E$ covered by boxes, depicted in yellow, within a $d$-dimensional hypercube projected to $2$ dimensions.}
    \label{fig:feasible-region-encoding-app}
\end{figure}

\begin{lemma}
    \label{lem:apx-feasible-region-boxes}
    Let $d \in \N$ and $u \in \{2,3\}$ be constant and let $k\in \N$ be constant with $d | k$.
    Given a $u$-uniform hypergraph $H = (V, E)$ with $|V| = n$, where $V$ is partitioned into $k$ color classes $V = V_1 \cup V_2 \cup ... \cup V_k$. 
    There is an encoding of $V_1 \times ... \times V_k$ to a $d$-dimensional hypercube so that for each $\bar e \in {V \choose u} \setminus E$ the feasible region of $\bar e$ can be covered by $O(n^{k/d-1})$ boxes, which can be computed in $O(n^{k/d-1})$ time.
\end{lemma}

\begin{proof}
\label{prf:feasible-region-boxes}
    For our proof we assume hypergraphs of uniformity $u=3$, the argument is analogous for $u=2$, i.e., graphs. 
    A slightly simplified example of the encoding and the forbidden regions of the non-edges can be seen in \Cref{fig:basic-reduction-idea-app}.
    By scaling, we assume that the $d$-dimensional hypercube has a side-length of $n^{k/d}$.
    Given the numbers $x_1, ..., x_{k/d} \in [n]$, we use the canonical function \[\ind: (x_1, ..., x_{k/d}) \mapsto x_1\cdot n^{k/d-1} + x_2\cdot n^{k/d-2} + ... + x_{k/d}\] to address an integer point in the range $[0, n^{k/d}-1]$ 
    We denote each unit hypercube as a cell, its relative origin can be addressed by $(\ind(X_1), ..., \ind(X_d))$ where $X_1, ..., X_d \in [n]^{k/d}$.

    We first give the encoding of $H$ to the hypercube which is the same as used in \cite{Chan23}. 
    We evenly assign the $k$ color-classes to the $d$ dimensions, i.e., the color classes $V_1, ..., V_{k/d}$ are assigned to the first dimension, the classes $V_{k/d + 1}, ..., V_{2k/d}$ to the second dimension, etc.
    Now a $k$-tuple $(v_1, ..., v_k) \in V_1 \times ... \times V_k$ is bijectively assigned to the cell at \[\cel((v_1,...,v_k)) = (\ind(v_1, ..., v_{k/d}), \ind(v_{k/d+1}, ..., v_{2k/d}), ..., \ind(v_{(d-1)k/d}, ..., v_{k})).\]

    For any non-edge $\bar e \in V^u \setminus E$, we want to cover the feasible region of $\bar e$ using (axis-aligned) boxes $\mathcal B$, as can be seen in \Cref{fig:feasible-region-encoding-app}.
    Assume that $\bar e = (v_i, v_j, v_\ell) \in V_i \times V_j \times V_\ell \setminus E$ for some $i<j<\ell \in [k]$ and let $d_i, d_j, d_\ell$ such that the color class $V_i, V_j, V_\ell$ was assigned to dimension $d_i, d_j, d_\ell$, respectively.
    We project the forbidden region of $\bar e$ separately to the dimensions $d_i, d_j, d_\ell$ so that we can cover the feasible, non-forbidden, region on this projection using intervals. 
    We extend the resulting intervals to the other $d-1$ dimensions, the resulting boxes, denoted as $\mathcal{B}_i, \mathcal{B}_j, \mathcal{B}_\ell$, respectively, are added to $\mathcal B$.
    More specifically, assume that $d_i$ is distinct from $d_j, d_\ell$. The boxes we add to $\mathcal B$ based on the projection to dimension $d_i$ are
    \begin{align*}
        \mathcal{B}_i =~& 
        \{      [0,n^{k/d}]^{d_i-1} \times [\ind(0, 0, ..., 0), \ind(0, ... 0, v_i, 0, ..., 0)) \times [0, n^{k/d}]^{d-d_i}, \\
        & [0,n^{k/d}]^{d_i-1} \times [\ind(0, ..., 0, v_i+1, 0, ..., 0), \ind(0, ..., 1, v_i, 0, ..., 0)) \times [0, n^{k/d}]^{d-d_i}, \\
        & \qquad \vdots \\
        & [0,n^{k/d}]^{d_i-1} \times [\ind(n, ..., n-1, v_i+1, 0, ..., 0), \ind(n, ..., n, v_i, 0, ..., 0) \times [0, n^{k/d}]^{d-d_i}, \\
        & [0,n^{k/d}]^{d_i-1} \times [\ind(n, ..., n, v_i+1, 0, ..., 0), \ind(n,..., n)] \times [0, n^{k/d}]^{d-d_i} \}.
    \end{align*}
    We see that $|\mathcal B_i| = n^{i \bmod k/d} + 1= O( n^{k/d - 1} )$ and these can be listed in time $O( n^{k/d - 1} )$.
    In case that $d_i$ is not distinct, we get fewer boxes that, when projected to $d_i$, include all of $\mathcal B_i$ in case $d_i$ is distinct: 
    Let $v_i$ be the only vertex of $\bar e$ being encoded into $d_i$, let $I$ be the set of indices such that there exists a cell whose $k$-tuple includes $\bar e$ and that in dimension $d_i$ has an index taken from $I$.
    Assuming that a second vertex $v_j$ of $\bar e$ is encoded in $d_i$, the analog set $\hat I$ is a subset of $I$ as all relevant indices have to match $\bar e$.

    We formally prove that $\mathcal B$ covers exactly the feasible region of $\bar e = (v_i, v_j ,v_\ell)$.
    Given a cell $C$ of the feasible region and let $\cel^{-1}(C) = (w_1, ..., w_k)$.
    As $C$ is within the feasible region, one of $w_i \neq v_i, w_j \neq v_j$, or $w_\ell \neq v_\ell$ must hold.
    Assume that w.l.o.g. $w_i < v_i$, the argument for $w_i > v_i$ or $w_j \neq v_j, w_\ell \neq v_\ell$ holds analogously. 
    Then $C$ is covered by the box $[0,n^{k/d}]^{d_i-1} \times [\ind(w_1, ..., w_{i-1}-1, v_i+1, 0, ..., 0), \ind(w_1, ..., w_{i-1}, v_i, 0, ..., 0)) \times [0, n^{k/d}]^{d-d_i}$ which is en element of $\mathcal B_i$ or, if $d_i$ is not distinct, is contained in another box in $\mathcal B_i$.
    
    Conversely, let $C$ with $\cel^{-1}(C) = (w_1, ..., w_k)$ be a covered cell, thus it is included in $\mathcal B_i, \mathcal B_j$, or $\mathcal B_\ell$. For it being included, it must have $w_i \neq v_i$, $w_j \neq v_j, w_\ell \neq v_\ell$, respectively, so $(w_1,..., w_k)$ does not include $\bar e$ and it belongs to the feasible region.

    As $\mathcal{B} = \mathcal B_i \cup \mathcal B_j \cup \mathcal B_\ell$ and $\mathcal B_i, \mathcal B_j, \mathcal B_\ell = O( n^{k/d - 1} )$, we have $\mathcal{B} =  O( n^{k/d - 1} )$.
\end{proof}

Using \Cref{lem:apx-feasible-region-boxes}, we can already turn a $k$-(hyper-)clique instance into a Translated (Box-)Shape instance with $n = O(|V|^u), m = O( |V|^{k/d - 1 + u})$ as we can treat each feasible region of the $O(|V^u|)$ non-edges as an independent shape with $O(|V|^{k/d - 1})$ boxes each. 

However, we wish to extend the reduction to different, more balanced, parameter choices.
Such a rebalancing of parameters would maintain the resulting lower bound for directed HuT if the product $n m$ stays the same.
Unfortunately, we are not able to prove such a result. 
However, by additionally decomposing the hypercube in an intermediate step into smaller parts, we get non-trivial lower bounds for directed HuT that apply for arbitrary distributions of $n,m$.

We first start with \Cref{lem:hypercube-decomposition-1D} allowing the decomposition of our $d$-dimensional hypercube into identical slices.
We directly continue with \Cref{lem:feasible-region-decomposition} generalizing this decomposition, so that we partition our hypercube into identical sub-regions.
Each of these can efficiently encode the feasible region of any non-edge, and we can collectively turn them into an instance of the Translated (Box-)Shape Problem with our desired parameter distribution.

\begin{lemma}
    \label{lem:hypercube-decomposition-1D}
    Let $\lambda \in [0,1]$, $u \in \{2,3\}$ be constant, and $d,k \in \N$ be constants with $d|k$.
    Let $H=(V, E)$ be a $u$-regular hypergraph. 
    For any non-edge $\bar e \in {V \choose u} \setminus E$ and for any dimension $d_i \in [d]$ we can decompose the hypercube where $H$ is encoded in, see \Cref{lem:apx-feasible-region-boxes}, into $O(|V|^{\lambda k/d})$ many slices in $d_i$ that are identical under translation, i.e., there exists a translation such that two slices match, including the feasible and forbidden regions of $\bar e$.
    Note that the slices will not necessarily be continuous but have at most $(|V|^{k/d})^{1-\lambda}$ many disconnected parts. The slices can be computed in time $O(|V|^{k/d})$. 
\end{lemma}

\begin{proof}
    We prove the statement by induction over the number of vertices of $\bar e$ encoded in $d_i$.

    Assume only one vertex of $\bar e$ is encoded in $d_i$, let this vertex be $v_i$ from the color-class $V_i$. As we have seen in \Cref{lem:apx-feasible-region-boxes}, in the projection of the hypercube to $d_i$ we can group the resulting intervals by whether some forbidden region of $\bar e$ is projected onto them or not. 
    These intervals are regular, because they depend on whether $v_i$ is encoded in this interval, which only depends on the index for $V_i$. 
    Thus, we can form up to $|V|^{\lambda k/d}$ many, continuous groups, each of them having the same size. 
    At the same time, these groups have the same sequence of unit intervals regarding whether a forbidden region of $\bar e$, projected to $d_i$, is mapped to the interval, or not.

    Otherwise, assume there is more than one vertex of $\bar e$ encoded in dimension $d_i$. 
    Let $V_i$ be the first color class encoded in $d_i$ that has a vertex $v_i$ in $\bar e$, i.e, there is no other color class $V_j$ that has a vertex in $\bar e$ with $j < i$ and $\lfloor j/(k/d) \rfloor = \lfloor i/(k/d) \rfloor$.
    We first consider the complete hypercube that in dimension $d_i$ only encodes the vertices from $V_{i+1}, ..., V_{d_ik/d}$, that is, we "remove" the first $i \bmod k/d$ many color classes from the encoding and disregard $v_i$ in $\bar e$. Thus dimension $d_i$ only has a range of $[0, |V|^{k/d - (i\bmod k/d)}]$. We denote it as the lower hypercube.
    By the inductive argument, we can decompose the lower hypercube into $(|V|^{k/d - (i \bmod k/d)})^\lambda$ many continuous slices in dimension $d_i$ that are identical under translation.
    Next, we consider the complete hypercube that in dimension $d_i$ only encodes the vertices from $V_{(d_i-1)k/d + 1}, ..., V_{i}$, which we denote as the higher hypercube. 
    There is only one vertex of $\bar e$ encoded in $d_i$ and by the induction basis, we can decompose the higher hypercube into $O((|V|^{i \bmod k/d})^\lambda)$ many slices.
    
    To combine these slices, we do the following: Each of the slices of the higher hypercube contains some intervals where a forbidden region of $\bar e$ is mapped to, called relevant intervals, and some intervals where this is not the case, called irrelevant.
    Each relevant interval spans the interval where the vertices from $V_{i+1}, ..., V_{d_ik/d}$ are mapped to, i.e., it spans a lower hypercube. 
    Thus, we can decompose a relevant interval further into $(|V|^{k/d - (i \bmod k/d)})^\lambda$ identical sub-slices using the decomposition of the lower hypercube.
    For an irrelevant interval, we can decompose it into $(|V|^{k/d - (i \bmod k/d)})^\lambda$ many continuous sub-slices. Note that these sub-slices are identical as well since there is no forbidden region mapped to it altogether.
    
    For each of the slices of the higher hypercube, we now construct $(|V|^{k/d - (i \bmod k/d)})^\lambda$ new slices. For the $j$-th such constructed slice, we append alternatingly the $j$-th sub-slice from the relevant / irrelevant intervals. As each of the relevant (irrelevant) sub-slices is identical to the other relevant (irrelevant) sub-slices, respectively, each newly constructed slice is identical. 
    Note that a newly constructed slice is thus split into at most $(|V|^{i \bmod k/d})^{1-\lambda} \leq (|V|^{k/d - 1})^{1-\lambda}$ many disconnected parts, as this is the number of different relevant/irrelevant intervals within a slice of the higher hypercube.
    In total, we create $(|V|^{i \bmod k/d})^\lambda\cdot (|V|^{k/d - (i \bmod k/d)})^\lambda = (|V|^{k/d})^\lambda$ many identical slices for dimension $d_i$, which fulfills the lemma's statement. 
\end{proof}

\begin{figure}
    \centering
    \includegraphics[width=\linewidth]{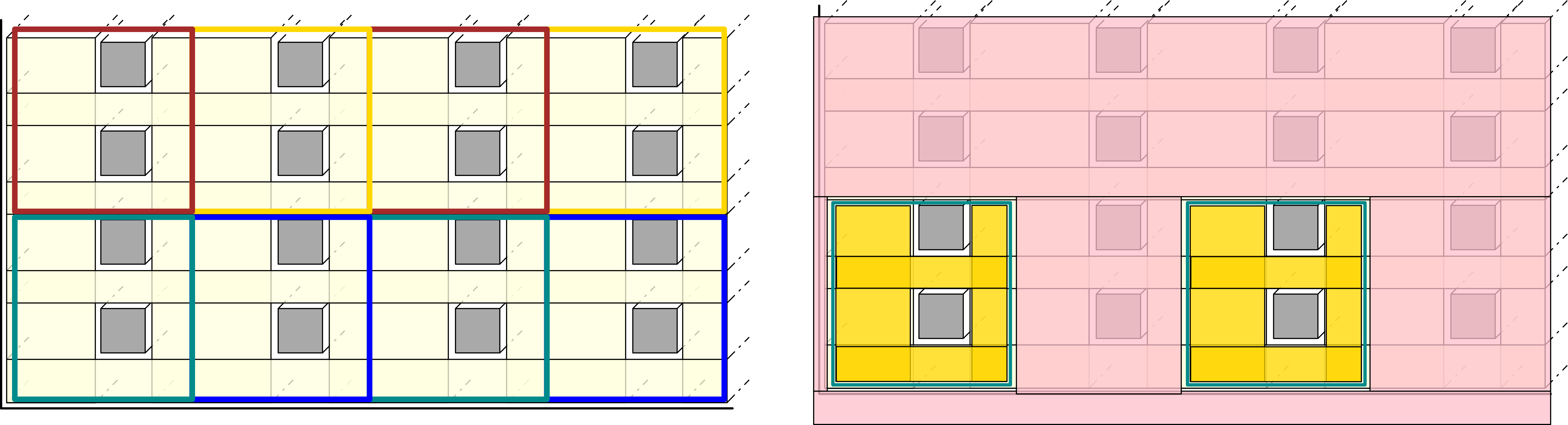}
    \caption{Left: The decomposition of a (skewed) hypercube in $2$-D into four identical sub-regions, marked in red, yellow, green, and blue. Boxes in equal color denote a common sub-region, here in each row there are two sub-regions that are split into two parts each.\\
    Right: The boxes used to cover the feasible region of the green sub-regions, in yellow, as well as the boxes covering the complement of the green sub-region, marked in pink.}
    \label{fig:hypercube-decomposition-app}
\end{figure}

\begin{lemma}
    \label{lem:feasible-region-decomposition}
    Let $\lambda \in [0,1]$, $u \in \{2,3\}$ be constant, and $d,k \in \N$ be constants with $d|k$.
    Given a $u$-regular hypergraph $H = (V,E)$ and its encoding to a $d$-dimensional hypercube as constructed in \Cref{lem:apx-feasible-region-boxes}.
    We reduce the problem of finding a $k$-(hyper-)clique in $H$ in $O(|V|^{k/d})$ time to an instance of the Translated (Box-)Shape Problem with $n = O(|V|^{u\lambda k/d + u})$ translates of common shapes consisting of $m = O(|V|^{(1-\lambda)k/d + u})$ boxes. 
\end{lemma}

\begin{proof}
\label{prf:feasible-region-decomposition}
    We consider the encoding of the hypergraph to the $d$-dimensional hypercube as given in \Cref{lem:apx-feasible-region-boxes}.
    As a $k$-(hyper-)clique is a $k$-tuple that does not include any non-edge, a cell that lies in the feasible regions of all non-edges encodes a $k$-(hyper-)clique, and vice versa.
    Thus, we have to show that we can encode the feasible regions of all non-edges in an instance of the Translated Shape Problem with the desired parameters.
    For this reduction, we show that the feasible region of any non-edge $\bar e$ can be represented by a Translated Shape instance with $O(|V|^{u\lambda k/d})$ translates of a shape with $O(|V|^{(1-\lambda)k/d})$ boxes.
    By \Cref{lem:shape-translation-composition}, we can compose the resulting shapes and translates of the $O(|V|^u)$ many non-edges into a single instance with the desired parameters.
    
    We first note that by adding a single translate and box, spanning exactly the hypercube we consider in \Cref{lem:apx-feasible-region-boxes}, i.e., $[0, |V|^{k/d}]^d$, only allows solutions to the Translated Shape Problem that lie within that hypercube. Thus, we can add boxes covering space outside our hypercube, without changing the resulting instance's feasibility. 

    Second, we see that a single non-edge is a $u$-tuple and thus its included vertices are only encoded in up to $u$ dimensions. 
    It suffices to consider the hypercube projected to these $u$ dimensions, as any cells from the feasible (forbidden) region can be extended to the other $d-u$ dimensions, while still being completely contained in the feasible (forbidden) region, respectively.
    
    For our construction, we first construct $((|V|^{k/d})^\lambda)^u$ many sub-regions that decompose the $u$-dimensional hypercube:
    By \Cref{lem:hypercube-decomposition-1D}, we can decompose each single dimension into $(|V|^{k/d})^\lambda$ many (not necessarily continuous) slices that are identical under translation, i.e., when translated so that both slices have the same relative origin, the feasible and forbidden regions of $\bar e$ within these slices match. 
    We combine these one-dimensional slices independently in all dimensions to get the desired decomposition into sub-regions.
    Note that the resulting sub-regions are indeed identical: Identical slices imply identical cells at the same relative position within a slice (belonging to the feasible or forbidden region of $\bar e$), so even if we disregard one dimension the given slices stay identical under translation.
    The number of sub-regions we construct thereby is $((|V|^{k/d})^\lambda)^u$.
 
    For the resulting Translated (Box-)Shape Problem we construct a single shape $S$ handling one sub-region. 
    We first add boxes to this shape that cover the complement of the sub-region, which are $2u$ many boxes to cover the space outside the sub-region's bounding box and at most $u(|V|^{k/d})^{1-\lambda}$ many boxes to cover the space between the disconnected parts, which is bounded by \Cref{lem:hypercube-decomposition-1D}. 
    Since all sub-regions are identical under translation and we have $(|V|^{k/d})^\lambda$ slices per dimension, the sub-regions have a (possibly discontinuous) side-length of $(|V|^{k/d})^{1-\lambda}$.
    Within a sub-region, \Cref{lem:apx-feasible-region-boxes} tells us that the feasible region of $\bar e$ constricted to the sub-region can be covered by $O(|V|^{(1-\lambda) k/d})$ many boxes. We add these boxes to our constructed shape.
    See \Cref{fig:hypercube-decomposition-app} for a visual example.
    We note that $|S| \leq O(|V|^{(1-\lambda) k/d})$. 

    For the translation objects, each translating the single constructed shape $S$, we use the translation vectors that mark the relative origin of each of the sub-regions, which are $((|V|^{k/d})^\lambda)^u$ many.

    For correctness, we see that each cell lies in exactly one sub-region. For a sub-region, we see that all cells that do not lie within this sub-region are covered. Thus, whether a cell lies within all translation objects solely depends on whether it lies within the boxes of its sub-region, which are constructed by \Cref{lem:apx-feasible-region-boxes}. As these boxes exactly cover the feasible region of $\bar e$, the solution set to our constructed Translation (Box-)Shape Problem instance equals the feasible region of our considered non-edge.
\end{proof} 
    
For our reduction to directed Hausdorff under Translation with a fixed parameter distribution we now assemble all necessary steps.
Note that since all parts also hold for graphs, i.e., hypergraphs of uniformity $2$, we not only get novel non-combinatorial but also combinatorial lower bounds. 
These combinatorial lower bounds are better for larger values of $\lambda$ compared to the non-combinatorial bounds as our balancing of parameters is susceptible to the uniformity of the input graph.

\begin{theorem}
    \label{thm:lopsided-directed-LB}
    Let $\varepsilon > 0, \lambda \in (0, 1], d \in \N_{\geq4}$ be constant. We get the following lower bounds.
    \begin{itemize}
        \item 
        Under the $k$-clique hypothesis, there is no $O((nm)^{\lfloor \frac{d}{2} \rfloor \frac{\lambda +2}{2\lambda+2} - \varepsilon})$ combinatorial algorithm for directed Hausdorff under Translation in $d$-D with $n = \Theta(m^\lambda)$.
  
        \item 
        Under the $3$-uniform $k$-hyperclique hypothesis, there is no $O((nm)^{\lfloor \frac{d}{2} \rfloor \frac{\lambda + 3}{3 \lambda + 3} - \varepsilon})$ (non-combinatorial) algorithm for directed Hausdorff under Translation in $d$-D with $n = \Theta(m^\lambda)$.
    \end{itemize}
\end{theorem}

\begin{proof}
\label{prf:lopsided-directed-LB}
    We here prove the second lower bound statement, where we reduce from the $3$-uniform $k$-hyperclique problem and $n =\Theta(m^\lambda)$. The other statement follow analogously, the necessary parameter for $\lambda'$ is given at the end of the proof.
    
    Assume $d$ to be even, otherwise we decrease $d$ by one for the following construction.
    Choose $\varepsilon' > 0$ and $k$ large enough so that $2k\varepsilon/d > 3d + \varepsilon'$, $d|k$, and
    \[\lambda' = \frac{\lambda}{3 + \lambda} - \frac{(3-3\lambda)d}{2k(3+\lambda)}>0.\]
    
    By \Cref{lem:apx-feasible-region-boxes,lem:feasible-region-decomposition}, we reduce the $3$-uniform $k$-hyperclique problem on a hypergraph with vertex set $V$ to the Translated Shape Problem in $d/2$-D with $n = \Theta(|V|^{3\lambda'k/(d/2)+3})$ translation objects and a total of $m = \Theta(|V|^{(1-\lambda')k/(d/2) + 3})$ boxes.
    It is easy to verify that $n = \Theta(m^\lambda)$.
    By \Cref{lem:shape-translation-reduction}, we reduce it to the Translation Problem with Boxes in $d/2$-D, which by \Cref{lem:box-to-cube} we further reduce to the Translation Problem with Orthants in $d$-D with input sizes $|P_C| = \Theta(n), |\mathcal C| = \Theta(m)$.
    The Translation Problem with Orthants directly reduces to the directed Hausdorff under Translation Problem, see \Cref{sec:translation-problems}, maintaining the parameter size.

    Given an algorithm running in time $O((nm)^{\lfloor \frac{d}{2} \rfloor \frac{\lambda + 3}{3 \lambda + 3} - \varepsilon})$ solving directed Hausdorff under Translation in $d$-D, we get an algorithm solving $3$-uniform $k$-hyperclique running in time
    \begin{align*}
        &\left(|V|^{3\lambda'k/(d/2) + 3} \cdot |V|^{(1-\lambda')k/(d/2) + 3}\right)^{\lfloor \frac{d}{2} \rfloor \frac{\lambda + 3}{3 \lambda + 3} - \varepsilon} \\
        = &~ \left( |V|^{2(1+2\lambda')k/d+6} \right)^{\frac{d}{2}  \frac{\lambda + 3}{3 \lambda + 3} - \varepsilon} \\
        = &~ \left( |V|^{2k/d + 4\left(\frac{\lambda}{3 + \lambda} - \frac{(3-3\lambda)d}{2k(3+\lambda)}\right)k/d+6} \right)^{ \frac{d}{2} \frac{\lambda + 3}{3 \lambda + 3} - \varepsilon} \\
        \leq &~ \left( |V|^{2k/d + 4\left(\frac{\lambda}{3 + \lambda}\right)k/d+6} \right)^{ \frac{d}{2} \frac{\lambda + 3}{3 \lambda + 3} - \varepsilon} \\
        \leq &~ \left( |V|^{k\frac{\lambda + 3}{3 \lambda + 3} - 2k\varepsilon/d}\right) \cdot 
            \left(|V|^{ 4\left(\frac{\lambda}{3 + \lambda} \right)( \frac{d}{2} \frac{\lambda + 3}{3 \lambda + 3}) k/d} \right)\cdot 
            \left(|V|^{6( \frac{d}{2} \frac{\lambda + 3}{3 \lambda + 3})} \right) \\
        \leq &~ \left( |V|^{k\frac{\lambda + 3}{3 \lambda + 3} - 2k\varepsilon/d}\right) \cdot 
            \left(|V|^{\left(\frac{2 \lambda}{3\lambda + 3} \right) k} \right) \cdot 
            \left(|V|^{3d} \right) \\
        = &~ \left( |V|^{k - 2k\varepsilon/d}\right) \cdot 
            \left(|V|^{3d} \right) \\
        < &~ |V|^{k - \varepsilon'},
    \end{align*}
    noting that we have chosen $k, \varepsilon'$ such that $2k\varepsilon/d > 3d + \varepsilon'$.

    For the combinatorial lower bound statement of this theorem, we set $\lambda' = \frac{\lambda}{2 + \lambda} - \frac{(2-2\lambda)d}{2k(2+\lambda)}$ accordingly.
    All other steps of the proof are exactly the same, using the respective statement from \Cref{lem:feasible-region-decomposition}.
\end{proof}

We note that \Cref{thm:lopsided-directed-LB} also yields (non-)combinatorial lower bounds for the case $m \in \Theta(n^\lambda)$ for $\lambda \in [0,1]$. These bounds however are always worse than the combinatorial $(nm)^{d/2 -o(1)}$ lower bound via a reduction from $k$-clique given in \cite{Chan23} and its respective non-combinatorial $(nm)^{d\omega/6 - o(1)}$ lower bound.

\subsection{Tight Non-Combinatorial Lower Bound in \texorpdfstring{$3$-D}{3-D}}
\label{sec:apx-HuT-3D-new}

In contrast to our lower bounds of \Cref{thm:lopsided-directed-LB}, Chan \cite{Chan23} has proven a tight lower bound that applies for distributions with $m \in O(n)$ in $3$-D.
We here give a non-combinatorial version of this $(nm)^{d/2-o(1)}$ lower bound that applies for $m \in O(n^{1/2})$.

However, his approach, reducing from $k$-clique, is not directly transferable to a reduction from $k$-hyperclique (\Cref{def:colorful-k-hyperclique}).
Instead, we will combine his ideas with the work of Künnemann \cite{kunnemann2022TightNoncombinatorialConditional} using \emph{prefix covering designs}.
As done in \Cref{sec:apx-lopsided}, we encode a $k$-tuple of vertices inside a $d$-dimensional hypercube. However, we use a redundant encoding, that is, each color class is encoded more than once to our hypercube.

We first prove that there exist so-called prefix covering sequences which later defines our encoding. 
While a similar encoding was already used by \cite{kunnemann2022TightNoncombinatorialConditional}, our analysis shows an additional crucial property. 
\begin{observation}[Prefix Covering Designs]
    \label{obs:good-pcd-designs}
	The sequences $s_1, s_2, s_3$, called \emph{prefix covering sequences}, defined over $\{1,...,k\}$ which we rename to $\{x_1, ..., x_{k/3}, y_1, ..., y_{k/3}, z_1, ..., z_{k/3}\}$, are of length $4k/9$ each and given as 
    \begin{align*}
        s_1 &= (x_1, x_2, ..., x_{k/3}, ~ y_{k/3}, ..., y_{2k/9+1}) \\
        s_2 &= (y_1, y_2, ..., y_{k/3}, ~ z_{k/3}, ..., z_{2k/9+1}) \\
        s_3 &= (z_1, z_2, ..., z_{k/3}, ~ x_{k/3}, ..., x_{2k/9+1}).
    \end{align*}
    For each $3$-tuple $e \in {k\choose 3}$, there exists $i, j, \ell \in \{0, ..., k\}$ such that there exists a (1) balanced prefix of $s_1, s_2, s_3$ of (2) combined length at most $2k/3+1$ that (3) covers $e$, formally,  
    \begin{enumerate}
        \item $\min(i,j,\ell) + \max(i,j,\ell) \leq 4k/9 $
        \item $i + j + \ell \leq 2k/3 + 1$
        \item $e \subseteq s_1[1..i] \cup s_2[1..j] \cup s_3[1..\ell].$
    \end{enumerate}
\end{observation}

\begin{proof}
    The properties (2), (3) have already been proven by \cite{kunnemann2022TightNoncombinatorialConditional}, however we reprove them for a complete picture.

    We use a case distinction on $e = (a, b, c)$.
    If $a,b,c \in \{x_1, ..., x_{k/3}, y_1, ..., y_{k/3}\}$, analog with $\{y_1, ..., z_1, ...\}$ and $\{z_1, ..., x_1, ...\}$, then $i = j = k/3, \ell = 0$ fulfills the definition.
    Otherwise, we can assume that, w.l.o.g., $a = x_i, b = y_j, c = z_\ell$.
    Assume $\ell \geq i,j$, the other cases are analog. 
    If $\max(i,j,\ell) \leq 2k/9$, the properties hold.
    Else, we set $j' = (2k/3 - \ell + 1)$, noting that $j' \leq 4k/9$, and see that $a,b,c \in s_1[1..i] \cup s_2[1..j']$. 
    Further, we have $i + j' = i + 2k/3 -\ell + 1 \leq 2k/3+1$.
    Thus, $i, j'$ and $\ell' = 0$ fulfill the properties.
\end{proof}

We next give a reduction from the $3$-uniform $k$-hyperclique problem to the Translated Orthant Problem, first introduced in \cite{Chan23}, which is a special case of the Translated Box-Shape Problem (\Cref{def:translated-box-shape}).

\begin{definition}[Translated Orthant Problem (Problem 4, \cite{Chan23})]
    \label{def:translated-orthant-shape}
    Let $\mathcal Z$ be a set of shapes in $\mathbb{R}^d$, composed of a union of orthants. Let $m$ be the total number of orthants over all shapes of $\mathcal Z$. 
    Given a set $\mathcal S$ of $n$ objects where each object is the translate of some shape in $\mathcal Z$, decide whether $\bigcap_{S \in \mathcal S} S \neq \emptyset$. 
\end{definition}

\begin{observation}
    \label{lem:orthant-translation-composition}
    Given two instances of the Translated Orthant Problem with inputs $\mathcal{S}_1, \mathcal{Z}_1$ and $\mathcal{S}_2, \mathcal{Z}_2$ of size $n_1, m_1$ and $n_2, m_2$ in a shared universe. 
    Deciding whether $\bigcap_{S \in \mathcal S_1 \cup \mathcal S_2}S \neq \emptyset$ reduces to a single instance of the Translated Orthant Problem with $\mathcal{S} = \mathcal{S}_1 \cup \mathcal S_2$, $\mathcal Z = \mathcal{Z}_1 \cup \mathcal Z_2$ of size $n \leq n_1 + n_2$ and $m \leq m_1 + m_2$.
\end{observation}
Analog to \Cref{lem:shape-translation-reduction} and first given by \cite{Chan23}, the Translated Orthant Problem with $n$ translation objects and a total of $m$ orthants is reducible to the Translation Problem with Orthants with $n$ points and $m$ orthants, and thus has a reduction to directed Hausdorff under Translation with point sets of size $n, m$. The proof can be found in \cite[Lemma 3]{Chan23}.

\begin{lemma}[Lemma 3, \cite{Chan23}]
    \label{lem:translated-orthant-reduction}
    The Translated Orthant Problem in $d$-D with $n$ translation objects and a total of $m$ orthants reduces to the Translation Problem with Orthants in $d$-D with $|P| = \Theta(n)$ points and $|\mathcal Q| = \Theta(m)$ orthants.
\end{lemma}

For our reduction, the encoding of $k$-tuples to our hypercube is given by the prefix covering sequences, i.e., the sequence $s_i$ defines the encoding into the $i$-th dimension of our $3$-dimensional hypercube. Note that this way, some vertex parts are encoded redundantly into the hypercube.
Again, we want to cover the feasible regions of each non-edge $\bar e$, that is the cells whose corresponding $k$-tuple does not include $\bar e$. Recall that the forbidden region of $\bar e$ are the cells for which the respective $k$-tuple includes $\bar e$, which basically is the complement of the feasible region\footnote{By the redundant encoding, we introduce regions that in our basic understanding neither lie in the forbidden nor feasible region. See the proof of \Cref{thm:pcd-reduction-3D} for details.}.
A cell lying in the feasible regions of all non-edges encodes a $k$-hyperclique.

Along the lines of the approach of Chan, to cover the feasible regions of each non-edge $\bar e$, we split the forbidden regions of $\bar e$ into \emph{quasi-diagonals}. A quasi-diagonal is a set of contiguous regions such that in all projections to a single dimension the projected regions are disjoint and feature the same ordering.
The complement of each quasi-diagonal in $3$-D of $x$ forbidden regions can be covered by $3x + 3$ orthants, a visual intuition as well as an example for a quasi-diagonal can be found in \Cref{fig:quasi-diagonal}. 

Besides giving the formal encoding and reduction to the Translated Orthant Problem, our main contribution lies in achieving the desired parameters for the resulting instance. 

\begin{lemma}
    \label{lem:forbidden-decomposition-quasi-diagonal}
    Given a $3$-uniform $k$-partite hypergraph $H=(V, E)$ with $n = |V|/k$ and $\lambda \in [2/3, 1]$.
    Given a non-redundant encoding of $H$ in a $3$-dimensional hypercube with side length $n^i \times n^j \times n^\ell$, for values that fulfill $\min(i,j,\ell) + \max(i + j +\ell) \leq 4k/9$ and $i + j + \ell \leq 2k/3 + 1$.
    For each $u$-tuple $\bar e \in {V \choose u} \setminus E$ for $u \in \{2,3\}$, the forbidden region of $\bar e$ can be decomposed into $O(n^{\lambda (2k/3 + 1)})$ many quasi-diagonals which are identical under translation. 
    Each of the quasi-diagonals has $O(n^{(1-\lambda)(2k/3+1)})$ many contiguous regions and intrinsic dimension $2$.
    The decomposition can be constructed in $O(n^{\lambda(2k/3+1)})$.
\end{lemma}

\begin{proof}
    \begin{figure}
        \begin{subfigure}[t]{0.45\textwidth}
            \centering
            \includegraphics[width=0.81\linewidth]{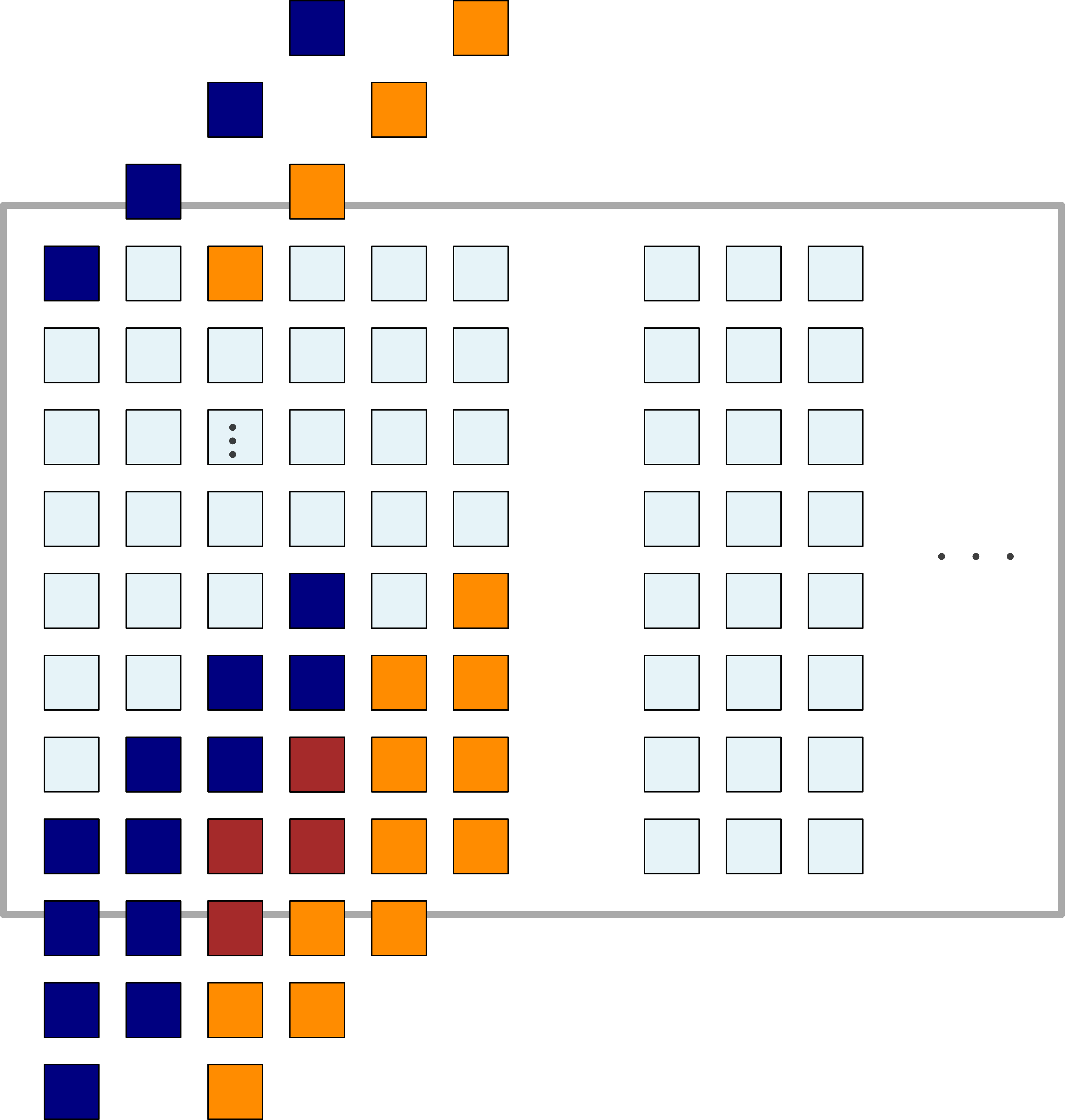}
            \subcaption{Case A, where a quasi-diagonal does not extend to another group of forbidden regions.}
        \end{subfigure}
        \hfill
        \begin{subfigure}[t]{0.45\textwidth}
            \centering
            \includegraphics[width=0.9\linewidth]{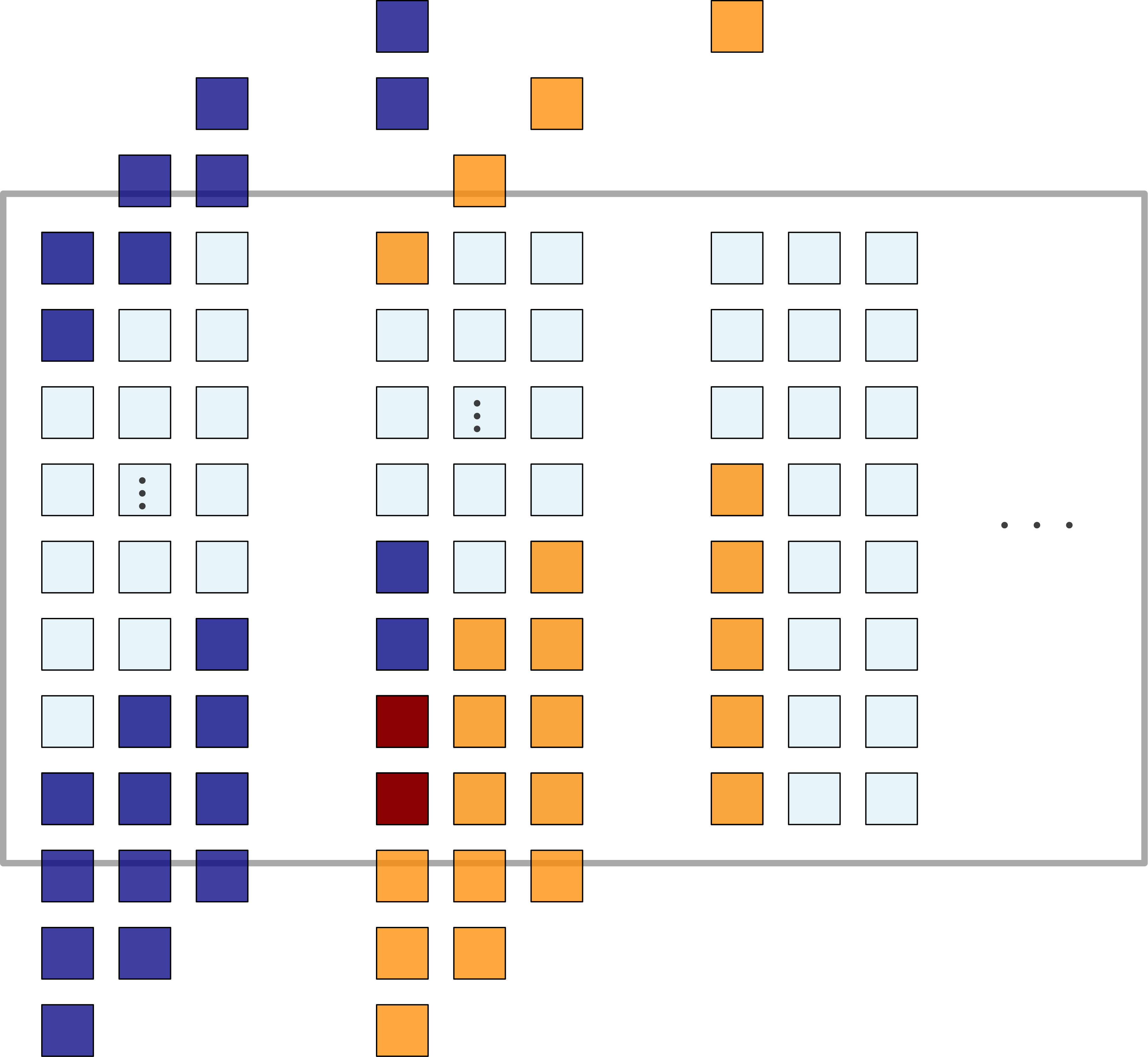}
            \subcaption{Case B, where a single quasi-diagonal crosses multiple groups of forbidden regions.}
        \end{subfigure}
        
        \caption{An schematic example of forbidden regions, marked in light blue, which are covered by quasi-diagonals, identical under translation, marked either in blue or orange. Overlapping blue and orange quasi-diagonals are depicted in brown.}
        \label{fig:quasi-diagonal-decomposition}
    \end{figure}
    Assume $i \leq j \leq \ell$, else we rename.
    We do a case distinction over the intrinsic dimension of (the forbidden regions of) $\bar e$, i.e., the number of dimensions $\bar e$ is encoded in.

    First, assume that $\bar e$ has intrinsic dimension $1$, i.e., $i = j = 0$ and $\ell > 0$.
    By $\max(i, j, \ell) \leq 4k/9$ and $\lambda(2k/3 + 1) > 4k/9$ we construct a one-element quasi-diagonal for each individual forbidden region of $\bar e$, which fulfills the definition.

    We next assume that $\bar e$ has intrinsic dimension $2$, that is $i = 0,$ $j > 0$ and we only need to consider the dimensions $d_1, d_2$. 
    Assume that only two elements are encoded in $d_1, d_2$, one in each dimension, thus the grid is regular in the sense that all neighboring contiguous regions have the same distance to each other.
    Indexing the $2$-D cells by their $2$-D position, we can form $O(n^\ell)$ many diagonals $D$ of the form $((x+y,x))_{x \in [n^j]}$ for each $y \in \{- n^j,..., n^\ell\}$. 
    Note that these also cover regions outside our hypercube, which we can ignore.
    As neighboring forbidden regions have the same distance to each other, we may split these diagonals further into $n^{\lambda(2k/3+1)-\ell}$ many parts which are identical under translation. 
    Note that $\ell = \max(i,j, \ell) \leq 4k/9 < 2/3 (2k/3+1) \leq \lambda(2k/3+1)$.
    Combining these yields $O(n^{\lambda(2k/3+1)})$ many quasi-diagonals.
    As we consider each forbidden region of $\bar e$ once and $j + \ell \leq 2k/3+1$, each slice has at most $O(n^{(1-\lambda)(2k/3+1)})$ many elements.

    Otherwise, we assume that there are two elements encoded in dimension $d_1$, having side-length $n^\ell$, and one in $d_2$, with side length $n^j$, the other case is analog.
    \Cref{fig:quasi-diagonal-decomposition} gives a visual overview on how we construct the quasi-diagonals.
    As before, we index the $2$-D cells by their $2$-D position.
    We initially form the diagonals of size $n^{\min(j,(1-\lambda)(2k/3+1))}$ 
    \[ D = \{((x,x + y))_{x \in [n^{\min(j,(1-\lambda)(2k/3+1))}]} \mid y \in \{- n^j,..., n^j\} \},\]
    depicted in blue in \Cref{fig:quasi-diagonal-decomposition}.
    Note that these diagonals are identical except for translation, as there is only one element encoded in dimension $d_2$ and thus all neighboring forbidden contiguous regions have the same distance in $d_2$.
    Note that we can decompose $d_1$ into identical slices, such that only one element is encoded in each slice, proven in \Cref{lem:hypercube-decomposition-1D} and visualized in \Cref{fig:quasi-diagonal-decomposition} as groups of forbidden regions.
    If the diagonals of $D$ do not fully cover one such slice, see Case A in \Cref{fig:quasi-diagonal-decomposition}, we add translates of $D$ within the same slice such that we fully cover the forbidden regions inside the slice. 
    We can do so naively, as within the slice there is only one element encoded in each dimension.
    We further add translates of $D$ to cover all other slices, which are identical, alike.
    Else, the diagonals of $D$ extend to other slices, as can be seen in Case B in \Cref{fig:quasi-diagonal-decomposition}.
    We add translates of D starting at the first not fully covered slice. As all slices are identical, the constructed translates cover the forbidden region in the next slice(s).
    Note that each forbidden region is covered at most twice, and for each column of forbidden regions $c$ there are at most twice as many quasi-diagonals with relative origin in column $c$ as there are forbidden regions in $c$. 
    Each quasi-diagonal has size $n^{\min(j,(1-\lambda)(2k/3+1))}$ and we have $[n^j] \times [n^\ell]$ forbidden regions, so by $i+j+\ell \leq 2k/3+1$ and $\ell \leq 4k/9 \leq \lambda2k/3$ we have created $O(n^{(j+\ell)/\min(j,(1-\lambda)(2k/3+1))}) \leq O(n^{\max(\ell, \lambda(2k/3+1))}) = O(n^{\lambda(2k/3+1)})$ many identical copies of quasi-diagonals.

    Lastly, we next move to the case where the forbidden region of $\bar e$ has intrinsic dimension $3$, i.e., $i >0$.
    We construct the decomposition into quasi-diagonals individually in each of the $n^{i}$ many slices in dimension $d_1$, note that these are identical by \Cref{lem:hypercube-decomposition-1D}.
    In the remaining dimensions $d_2, d_3$, there is one vertex encoded in each, thus by our previous argument, we get $O(n^\ell)$ many translations of a single quasi-diagonal of size $n^j$.
    We may split this diagonal further into $n^{\lambda(2k/3+1)-i-\ell}$ many equal parts, note that $i + \ell \leq 4k/9 < 2/3 (2k/3+1) \leq \lambda(2k/3+1)$.
    Overall in all $n^i$ slices in dimension $d_1$, we get $O(n^{\lambda(2k/3+1)})$ many translations of our diagonal. Again, we unnecessarily cover each forbidden region only constantly often, so by $i + j + \ell \leq 2k/3+1$ our diagonal has size $O(n^{(1-\lambda)(2k/3+1)})$.

    For all cases, to determine the first quasi-diagonal and subsequent translations thereof, we only need to know the grid's structure. Thus, our runtime is dominated by outputting the quasi-diagonal and its translations.
\end{proof}

\begin{figure}
    \centering
    \includegraphics[width=0.5\linewidth]{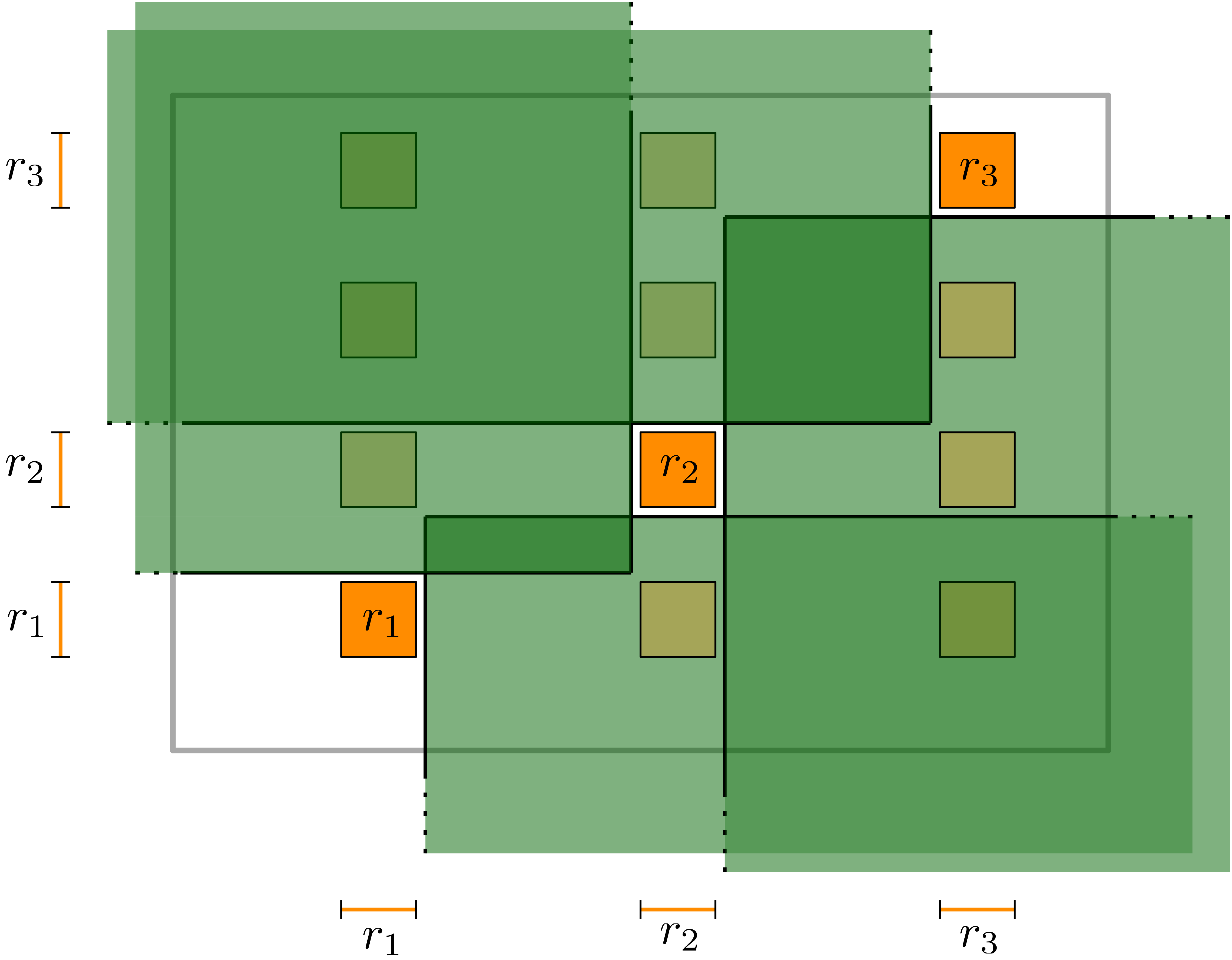}
    \caption{A schematic picture of a quasi-diagonal and its complement-covering orthants.
    The (faint) orange boxes depict the forbidden regions of a non-edge $\bar e$, the strongly colored boxes form a quasi-diagonal as their projections have the same ordering. The green orthants are the ones constructed in \Cref{thm:pcd-reduction-3D}.}
    \label{fig:quasi-diagonal}
\end{figure}

\begin{theorem} 
	\label{thm:pcd-reduction-3D}
    Let $\lambda \in [2/3,1]$ and $k \in \N$ with $3|k$ be constants.
    The $k$-hyperclique problem on a $3$-regular graph with $|V|$ vertices reduces to the Translated Orthant Problem in $3$-D with $n' = \Theta(|V|^{\lambda 2k/3 + 4})$ points and a total of $m' = \Theta(|V|^{(1-\lambda)2k/3 + 4})$ cubes.
    The reduction runs in $O(|V|^{\lambda2k/3 + 4})$ time. 
\end{theorem}

\begin{proof}
    We reduce from the equivalent colorful $k$-hyperclique problem on $3$-regular hypergraphs $H = (V, E)$, that is, the vertex set is decomposed into $k$ disjoint parts $V = V_1 \cup ... \cup V_k$, each with $n$ vertices. Note that $n \in \Theta(|V|)$.
    We start by giving the encoding of $k$-tuples $C \in V_1 \times ... \times V_k$ to the $3$-dimensional hypercube\footnote{For consistency with previous proofs, we call it the hypercube although being $3$-dimensional.} and continue constructing the orthants covering the feasible regions of a non-edge.    

    By scaling, we assume that the $3$-dimensional hypercube has a side-length of $n^{4k/9}$.
    We use the prefix covering sequences $s_1, s_2,s_3$ of \Cref{obs:good-pcd-designs} to define the encoding of the $k$-tuples to the hypercube. 
    Let $\ind_i: V_1 \times ... \times V_k \to [n^{4k/9}]$ be the function which assigns a $k$-tuple to a unit interval in dimension $i\in[3]$ with \[
        {\ind}_i: (x_1, ..., x_{k}) \mapsto x_{s_i[1]}\cdot n^{4k/9-1} + x_{s_i[2]}\cdot n^{4k/9-2} + ... + x_{s_i[4k/9]}.
    \]
    Likewise, we define the assignment function $f: V_1 \times ... \times V_k \to [n^{4k/9}]^3$, assigning each $k$-tuple a unit cell in the hypercube with $f(C) = (\ind_1(C), \ind_2(C), \ind_3(C))$.
    
    First, for each non-edge $\bar e \in {V \choose 3} \setminus E$, we construct an instance of the Translated Orthant Problem $I$ such that exactly the feasible region of $\bar e$ is covered by the translation objects of $I$. 
    Second, we cover the complement of so-called inconsistent cells:
    As we use up to two indices to encode each vertex part $V_{x_{2k/9+1}}, ..., V_{x_{k/3}}, V_{y_{2k/9+1}}, ..., V_{y_{k/3}}, V_{z_{2k/9+1}}, ..., V_{z_{k/3}}$, there are cells which for the same vertex part encode different vertices, which we call inconsistent as they do not encode a $k$-tuple of $V_1 \times ... \times V_k$.
    For each of the mentioned vertex parts $V_{e_i}$, we construct an additional Translated Orthant instance which covers the complement of the cells which are inconsistent in the indexing of $V_{e_i}$.
    By \Cref{lem:orthant-translation-composition}, we can combine all individual instances into a single compound one.

    We start to construct an instance covering the feasible region of a non-edge $\bar e \in {V \choose 3} \setminus E$.
    Let $\bar e = (v_i, v_j, v_\ell) \in V_{e_i} \times V_{e_j} \times V_{e_\ell}$, by the definition of the prefix covering sequences, there exist $i,j, \ell$ such that $e_i, e_j, e_\ell \in s_1[1..i] \cup s_2[1..j] \cup s_3[1..\ell]$.
    We note that the forbidden region of $\bar e$ can be fully described in the coarsement of our hypercube defined by the assignment function $f(X) = (\ind^i_1(X), \ind^j_2(X), \ind^\ell_3(X))$ with \[{\ind}_z^i: (x_1, ..., x_{k}) \mapsto x_{s_z[1]}\cdot n^{4k/9-1} + ... + x_{s_z[i]}\cdot n^{4k/9-i},\] for $z \in \{1,2,3\}$.
    The resulting coarse hypercube has scaled side length of $n^i \times n^j \times n^\ell$.
    By \Cref{lem:forbidden-decomposition-quasi-diagonal}, we can decompose the forbidden region of $\bar e$ into $O(n^{\lambda (2k/3 + 1)})$ many quasi-diagonals of size $m = O(n^{(1-\lambda)(2k/3 + 1)})$, which are identical except for translation, each having intrinsic dimension $2$.
    Note that having intrinsic dimension $2$ implies that all regions in the two intrinsic dimensions are extended to the complete third dimension, which we denote by $[\infty]$. 
    
    We construct the Translated Orthant instance with shapes $\mathcal{Z}$ and translation objects $\mathcal S$.
    For a single quasi-diagonal with intrinsic dimensions $d_1, d_2$ and origin $q$, we construct a single shape $Z$ covering the complement of the quasi-diagonal: Let $r_i, r_{i+1}$ be neighboring elements in the quasi-diagonal with $r_i \prec r_{i+1}$. We add the orthants $(-\infty, (r_{i+1})_{d_1}] \times [(r_i)_{d_1}, \infty) \times [\infty]$ and $[(r_i)_{d_2}, \infty) \times (-\infty, (r_{i+1})_{d_2}] \times [\infty]$, a visual example is given in \Cref{fig:quasi-diagonal}.
    We additionally add orthants covering the space before the first ($r_1$) and after the last ($r_\ell$) element of the quasi-diagonal, that is we add $(-\infty, (r_1)_{d_1}] \times [\infty] \times [\infty]$, $[\infty] \times (-\infty, (r_1)_{d_2}] \times [\infty]$ and $[(r_\ell)_{d_1}, \infty) \times [\infty] \times [\infty]$, $[\infty] \times [(r_\ell)_{d_2}, \infty) \times [\infty]$.
    Clearly, $Z$ covers exactly the complement of the quasi-diagonal and consists of $m + 4$ many orthants.
    For each of the $O(n^{\lambda (2k/3 + 1)})$ quasi-diagonals with origin $q'$, we add a translation object translating $Z$ by $(q' - q)$.
    As the quasi-diagonals are identical when translated by $(q' - q)$, the translated shape of $Z$ covers exactly the complement of the quasi-diagonal with origin $q'$.

    Next, we cover the complement of the inconsistent cells for a redundantly encoded vertex part $V_{e_i}$, with $e \in \{x,y,z\}$, $i \in [2k/9+1,k/3]$.
    Handling the redundant indexing of $V_{e_i}$ as a new vertex part $V_{e_i}'$, we see that the inconsistent regions corresponding to $V_{e_i}$ can be represented as $O(n^2)$ many $2$-uniform non-edges of the form $(x, x') \in V_{e_i} \times V_{e_i}'$ with $x,x' \in [n]$ and $x \neq x'$.
    Let $j, \ell$ be the indices at which $s_z[j] = e_i, s_{z'}[\ell] = e_i'$ for some $z,z' \in \{1,2,3\}$.
    As before, we can coarsen our hypercube to be of dimension $n^{j} \times n^{\ell} \times n^0$.
    By the definition of the prefix covering sequences, we have $j + \ell \leq 2k/3+1$.
    We can thus apply \Cref{lem:forbidden-decomposition-quasi-diagonal}, yielding a decomposition of the forbidden, here inconsistent, regions into $O(n^{\lambda (2k/3 + 1)})$ many quasi-diagonals of size $m = O(n^{(1-\lambda)(2k/3 + 1)})\leq O(n^{(1-\lambda)2k/3 + 1})$.
    By the same construction as for the forbidden regions of non-edges, we can cover the complement of the quasi-diagonals by a single instance of the Translated Orthant Problem with one shape consisting of $m + 4$ orthants and $O(n^{\lambda (2k/3 + 1)}) \leq O(n^{\lambda2k/3 + 1})$ translation objects.
    Doing so for all $n^2$ constructed non-edges, we get a combined instance with $O(n^{\lambda 2k/3 + 3})$ translation objects and a total of $O(n^{(1-\lambda)2k/3 + 3})$ orthants.

    For the runtime of this construction, note that for each non-edge we can construct the instance of the Translated Orthant Problem in time $O(n^{\lambda(2k/3+1)}) \leq O(n^{\lambda2k/3+1})$ by \Cref{lem:forbidden-decomposition-quasi-diagonal}. Doing so for all $O(n^3)$ non-edges takes time $O(n^{\lambda2k/3 + 4})$.
    The decomposition of the inconsistent regions into quasi-diagonals and the construction of the complement-covering instances of the Translated Orthant Problem also runs in $k/3 ~n^{\lambda2k/3 + 3} = O(n^{\lambda2k/3 + 3})$ time.

    For correctness, all inconsistent regions are not covered by at least one of the Translated Orthant instances.
    For the consistent regions, for each non-edge, we cover the feasible region of at least one (of possibly multiple redundant) of the non-edge's encodings. As a cell corresponding to a $k$-tuple $C$ is only consistent if it includes all encodings of $C$, a cell lying in all feasible regions encodes a $k$-tuple which does not include a non-edge. Conversely, any $k$-hyperclique does not include any non-edge and is thus in the intersection of all Translated Orthant instances.
\end{proof}

\begin{corollary}
    \label{cor:noncomb-LB-3D}
    Let $\varepsilon > 0$.
    For any $\lambda \in (0, 1/2]$, there is no $O((nm)^{3/2 - \varepsilon})$ (non-combinatorial) algorithm for directed Hausdorff under Translation in $3$-D with $m = \Theta(n^\lambda)$ under the $k$-hyperclique hypothesis.
\end{corollary}

\begin{proof}
    Choose $k$, $\varepsilon' > 0$ such that $2k\varepsilon/3 \geq 8(3/2-\varepsilon) + \varepsilon'$ holds and $\lambda'= \frac{1}{1+\lambda} - \frac{12(\lambda-1)}{2k(1+\lambda)} > 0$.
    By \Cref{thm:pcd-reduction-3D}, we reduce $3$-uniform $k$-hyperclique on $V$ vertices to the Translation Problem with Orthants, which has a reduction to directed Hausdorff under Translation, with $n = \Theta(|V|^{\lambda'2k/3 + 4})$ and $m = \Theta(|V|^{(1-\lambda')2k/3 + 4})$, which implies $m = \Theta(n^\lambda)$.

    A $O((nm)^{3/2 - \varepsilon})$ algorithm for directed Hausdorff under Translation in $3$-D for $m = \Theta(n^\lambda)$ thus gives an algorithm for $3$-uniform $k$-hyperclique in time \[
    O\left(|V|^{\lambda'2k/3 + 4} |V|^{(1-\lambda')2k/3 + 4}\right)^{3/2-\varepsilon} = O(|V|^{k - 2k\varepsilon/3 + 8(3/2-\varepsilon)}) \leq O(|V|^{k-\varepsilon'}).
    \]
\end{proof}

It is possible to extend this approach to higher dimensions. Unfortunately, the obtained lower bounds introduce a gap between lower and upper bound, increasing with $d \geq 4$, and only apply to smaller values $\lambda$ so that $m \in \Theta(n^\lambda)$.

Prefix covering designs for related problems in higher dimension were studied by Gorbachev and Künnemann \cite{gorbachev_combinatorial_2023}. 
They showed that for $d \geq 4$ no sequences $s_1, ..., s_d$ exist so that any $3$-tuple lies in a combined prefix of size at most $2k/d + 1$ (properties (2), (3) of \Cref{obs:good-pcd-designs}). As a result, if we use this approach of prefix covering designs, we are not able to achieve tight lower bounds. However, they introduce a parameter $\gamma_d$, which essentially denotes the best exponent provable via a direct application of prefix covering designs. For directed Hausdorff under Translation in $d$-D using prefix covering designs as presented here, we can obtain a non-combinatorial lower bound of $(nm)^{\gamma_d - o(1)}$. (The values for $\gamma_d$ can be found in \cite{gorbachev_combinatorial_2023}; notably, for $d \geq 4$ they prove $\gamma_d < d/2$.)

These bounds also only apply for small, polynomial-sized, $m \ll n$ since we cannot bound the maximum length of the prefix covering design covering a non-edge.
While there is some (probably small) room for improvement to increase the ratio of $m$ to $n$, we leave out a further analysis of these cases for $d \geq 4$. 

\begin{remark}
    \label{rem:pcd-reduction-all-D}
    Let $\varepsilon > 0$, $d\geq 4$.
    There is no $O((nm)^{\gamma_d - \varepsilon})$ (non-combinatorial) algorithm for directed Hausdorff under Translation in $d$-D with $m = O(n^{\varepsilon/2})$ under the $3$-uniform $k$-hyperclique hypothesis.
\end{remark}

\subsection{A Faster Algorithm for \texorpdfstring{$m \gg n$}{m >> n}}
\label{sec:apx-other-side}
For $n = m^{o(1)}$, \Cref{thm:lopsided-directed-LB} gives a $(nm)^{\lfloor d/2 \rfloor - o(1)}$ lower bound, which for even $d$ matches the best known algorithm running in time $O((nm)^{d/2})$ \cite{Chan23}, while for odd $d$ it seems to have some slack.
Intuitively, one might expect that the known $(nm)^{d/2 - o(1)}$ combinatorial lower bound for $m = O(n)$ \cite{Chan23} also transfers to the case of $m \gg n$.
Surprisingly, in $3$-D the $(nm)^{1 - o(1)}$ lower bound of \cite{BringmannN21} is tight for $n = m^{o(1)}$ as we give an $O(m^{1 + o(1)})$ algorithm.

The algorithm uses a structural insight from~\cite[Theorem 7]{chew1999GeometricPatternMatching}.
Given a directed HuT instance with points $P$ and hypercubes $\mathcal Q$ in $d$-dimensions, they already gave the combinatorial proof that the region of feasible translations, i.e., the region of translation vectors for which all points $P$ are translated into the hypercubes of $\mathcal Q$, has $O(m^{\lfloor d/2 \rfloor}n^d)$ many vertices. 
However, their proof was not algorithmic yet, as it needed the explicit construction (of the vertices) of the union of $O(m)$ unit hypercubes in time $O(m^{\lfloor d/2 \rfloor})$. 
To our knowledge, such an algorithm is not yet known for arbitrary odd $d$. 
For the case $d = 3$ though, such an algorithm has recently been introduced by Agarwal and Steiger~\cite{agarwal_output-sensitive_2021}, which computes the vertices of the union of $n$ unit cubes in an output-sensitive runtime of $\tOh(n)$. 
We can thus turn the combinatorial observation into an algorithm for the directed Hausdorff under Translation problem in $3$-D, additionally filtering vertices that do not correspond to valid translations.

\begin{theorem}
    \label{thm:lopsided-algo-3D}
    Directed Hausdorff under Translation in $3$-D can be solved in time $\tOh(n^4m)$.
\end{theorem}

\begin{proof}

    Given an instance of the directed Hausdorff under Translation Problem in $d=3$ dimensions with point sets $P$ and $Q$, for which we assume general position.
    
    Let $\mathcal{L} = \{L_1, ..., L_n\}$ be a collection of unions of unit hypercubes, called layers, where $L_i = \{q + [-\delta, \delta]^d - p_i \mid q \in Q\}$, i.e., the region of translations feasible for $p_i \in P$. To decide whether the instance has a solution, it would be sufficient to test all vertices of $\bigcap \mathcal{L}$ as translation vectors. While doing so directly, would be too costly, the following algorithm achieves this efficiently.
   
    For each set of at most $d$ many layers\footnote{Since our arguments carry over to $d>3$, we write $d$ rather than 3 in the following proof and corollary.} $\mathcal L' \subset \mathcal L$, we compute the union of its layers, $\bigcup_{L' \in \mathcal L'} L'$, and its resulting vertices $V(\mathcal L')$. 
    Let $V = \bigcup V(\mathcal L')$ be the union over all these vertex sets.
    For each $v \in V$, we check whether it is a feasible translation for all points $p_i \in P$, meaning that there is a point $q \in Q - p_i$ in distance $|q - v|_\infty \leq \delta$.
    To do that, we precompute a $d$-dimensional range tree $\RT_i$ for each\footnote{Alternatively, we could compute a range tree for $Q$ and query with $v + p_i + [-\delta, \delta]^d$.} point set $Q - p_i$. 
    For each $v \in V$ and range tree $\RT_i$ we perform an orthogonal range query $v + [-\delta, \delta]^d$.
    If for each range tree, we find a point in the queried range, we return $v$ as a feasible translation.
    If this procedure does not yield a result, we return NO.

    Computing the range trees takes $O(n \cdot m \log^{d-1}m)$ time while queries run in $O(\log^{d-1}m)$ \cite{AlstrupBR00}.
    We have $O(n^d)$ many configurations of at most $d$ layers, each of these configurations consists of $dm = O(m)$ many unit hypercubes.
    For $d = 3$, computing the vertices of the union of these hypercubes takes time $\tOh(m^{\lfloor d/2 \rfloor})$ by an algorithm from \cite{agarwal_output-sensitive_2021}. 
    For each of the resulting $\tOh(n^dm^{\lfloor d/2 \rfloor})$ vertices it takes $\tOh(n)$ time to query all range trees. 
    Setting $d = 3$ gives the stated runtime.

    For correctness, consider a vertex $v$ of the intersection of $\bigcap \mathcal L = \bigcap_{L \in \mathcal L} L$. 
    Chew, Dor, Efrat and Kedem~\cite{chew1999GeometricPatternMatching} previously observed that $v$ also is a vertex of the union of $\bigcup_{L \in \mathcal L'} L$ for some $\mathcal L' \subseteq \mathcal L$ with $|\mathcal L'| \leq d$:
    Given $v$, let $\mathcal L'$ be the inclusion-minimal subset of layers $\mathcal L' \subset \mathcal L$ such that $v$ is a vertex of $\bigcap \mathcal L'$, i.e., there is no $\mathcal L'' \subset \mathcal L'$ such that $v$ is a vertex of $\bigcap\mathcal L''$.
    For the sake of contradiction, assume that $v$ is not a vertex of the union $\bigcup \mathcal L'$ but lies in the interior of some $L \in \mathcal L'$ as we have assumed general position. 
    Then, $v$ is also a vertex of $\bigcap (\mathcal L' \setminus \{L\})$ and $\mathcal L'$ is not inclusion-minimal.
    Note that such an inclusion-minimal subset may contain at most $d$ layers.

    As we construct all subsets $\mathcal L' \subseteq \mathcal L$ with $|\mathcal L'| \leq d$ and list all vertices of $\bigcup \mathcal L'$, we construct a set including all vertices of the intersection $\bigcap \mathcal L$.
    A solution to the directed HuT instance equivalently lies within $\bigcap \mathcal L$, so iff this region is non-empty, a vertex of this intersection is a feasible translation vector to the directed HuT instance.
    As we check for each constructed vertex whether it is a feasible translation vector, we must find a solution iff there exists one.
\end{proof}

The algorithm directly generalizes to higher dimensions in time $\tOh(m^{\lfloor d/2 \rfloor} n^{d+1})$ assuming an -- unfortunately still unknown -- algorithm that computes the vertices of the union of $m$ unit hypercubes in $d$-D in time $\tOh(m^{\lfloor d/2 \rfloor})$.

\Cref{thm:lopsided-algo-3D} also shows that a $(nm)^{d/2 -o(1)}$ lower bound for arbitrary distributions of $n,m$ does not exist for $d=3$, and we give compelling intuition for a similar result in odd $d>3$.
More specifically, there cannot be a $(nm)^{1 + \varepsilon - o(1)}$ lower bound with constant $\varepsilon > 0$ that applies for $n = O(m^{\lambda})$ with an arbitrary constant $\lambda > 0$.

\begin{corollary}
    \label{cor:d/2-floor-algo}
    Let $d=3$, $\kappa > 0$ and set $\lambda = \frac{\kappa}{d+2}$.     
    For $n = O(m^{\lambda})$, directed Hausdorff under Translation in $d$-D can be solved in time $O(m^{\lfloor d/2\rfloor}n^{d+1}) < O(m^{\lfloor d/2\rfloor +\kappa})$.
\end{corollary}

Thus, our lower bound from \Cref{thm:lopsided-directed-LB} is tight for $n = m^{o(1)}$.
Anyway, there is a tight (combinatorial) lower bound of $(nm)^{d/2 -o(1)}$ for $m \in O(n)$ and arbitrary $d \geq 3$, which points to a discrepancy in complexity. 
It begs the question of where the cutoff point is located, i.e., for which values of $n,m$ do we start to get faster than $(nm)^{d/2\pm o(1)}$ algorithms?

\section{Connections to Additive Problems}
\label{sec:apx-additive-problems}
As Hausdorff under Translation has been studied from the angle of additive problems \cite{gokaj_completeness_2025}, we place it in relation to different additive problems. While we mostly focus on discrete HuT, we also show hardness results for the continuous variant.

We organize this section as follows. First, we give lower bounds for directed Hausdorff under Translation in $1$-D in both the discrete and continuous variant, giving evidence that the difference seen in the algorithm runtimes for the directed and undirected case (see \Cref{tab:results-d1-d2}) stems from an intrinsic difference in complexity.

In the second subsection, we consider the discrete variant, which is known to be complete for the class of $\FOPZ(\exists\forall\exists)$ problems \cite{gokaj_completeness_2025}, meaning that all formulas in the class admit a (possibly high-dimensional) fine-grained reduction to a discrete HuT instance. 
Roughly speaking, the class $\FOPZ(\exists\forall\exists)$ contains all formulas of the type $\exists a \in A ~\forall b \in B~ \exists c \in C : \varphi(a,b,c)$, where $\varphi$ is a quantifier free linear arithmetic formula over the multidimensional vectors $a,b,c$; for a precise definition see \Cref{sec:definitions}.

We further extend the progress on class-based lower bounds by showing that discrete HuT is hard for all problems $\mathsf{FOP}_{\mathbb{Z}}(\forall \exists \exists)$ which can be expressed by using at most $d \leq 3$ inequalities (i.e., are of inequality dimension $d \leq 3$). 
In the context of upper bounds in dimension $d \leq 3$, we provide a reduction to 3-SUM for the balanced case and a slight variant of $3$-SUM called All-Ints $3$-SUM$(n,m,k)$ for the unbalanced case. 

\subsection{\texorpdfstring{Lower Bounds for $d = 1$}{Lower Bounds for d = 1}}
\label{sec:lower-bounds-1D}
We first focus on the one-dimensional case. 
While we have an almost-linear algorithm for the undirected HuT problem \cite{rote1991ComputingMinimumHausdorff}, which directly transfers to undirected DiscHuT as it constructs the complete region of feasible translations (see \Cref{sec:discHuT-1D-linear} for more details), our known algorithms for directed HuT have a running time of $\tOh(nm)$ \cite{huttenlocher1990ComputingMinimumHausdorff}. 
For Discrete HuT, one obtains additional straightforward algorithms running in time $\tOh(t\cdot\min(n,m))$, shortly described in \Cref{sec:baseline}, which in the balanced case turns to $\tOh(n^2)$.

Lower bounds for $1$-dimensional (directed) HuT seem particularly difficult to achieve as the single dimension allows little degree of freedom that is necessary for higher-dimensional lower bounds.  
Therefore, we turn our attention to additive problems. Cygan, Mucha, Wegrzycki and Wlodarczyk~\cite{cygan2019ProblemsEquivalentMin+Convolution} studied several variants of $(\min,+)$-Convolution, where they used MaxConv LowerBound as hardness barriers for giving sub-quadratic algorithms for $L_\infty$-Necklace Alignment.
Specifically, they showed a subquadratic reduction from MaxConv LowerBound to $L_\infty$-Necklace Alignment, which in turn has a subquadratic reduction to $(\min,+)$-Convolution \cite{BremnerCDEHILPT14}.

We also use these problems for lower bounds for (discrete) HuT. 
At first, we reduce from the Linear Alignment problem, which is known to be harder than the $L_\infty$-Necklace Alignment problem; see \cite{BremnerCDEHILPT14}.

\begin{definition}[Linear Alignment problem]
    Given sorted vectors $A,B$ of real numbers of size $n$ and $m$ respectively, where $m \geq n.$ 
    Compute an integer $s\leq m-n$ and a real-valued $c \in [0,1)$, which minimize the value 
    $\max_{i\in[n]} \left(\left | \left(A[i] +c \right) - B[i+s]  \right | \right).$
\end{definition}

\begin{theorem}
    \label{thm:necklace-LB}
    Let $P, Q$ be the given point sets of a $1$-dimensional HuT instance. 
    If we can compute a translation $\tau \in \R$ that minimizes $\hddist(P + \tau, Q)$ in time $T(|P|,|Q|)$, then the Linear Alignment problem can be solved in time $\tilde{O}(T(|A|,|B|) +\max(|A|,|B|)).$
\end{theorem}

\begin{proof}
    Let $n,m$, with $m \geq n$, denote the length of $A$ and $B$, respectively, which is the input of the Linear Alignment problem. 
    Moreover, let $M$ be a sufficiently large number, e.g., $M=10 nm \cdot \max_{x \in A \cup B}|x|$. We construct the sets
    \begin{align*}
        P&= \{A[i]+ iM \mid i \in [n] \},\\
        Q&= \{B[i] +iM \mid i \in [m] \}.
    \end{align*}
    
    To show the correctness of our reduction it suffices to show that 
    \begin{align*}
        \min_{\tau \in \mathbb{R}}\max_{p \in P} \min_{q \in Q} \left( |(p +\tau)-q| \right) = \min_{s \in \Z,c \in \mathbb{R}} \max_{i \in [n]} \left( \left| \left( A'[i]+c \right) - B'[i+s] \right| \right).
    \end{align*}
    
    Assume that $\tau:= sM+c$ for $s\in \mathbb{Z}$, $c\in [0,M)$ is chosen such that the left-hand side is minimized, being equal to value $\delta^*$. 
    We know that for $\tau = 0$ the left-hand side of the equation is at most $M$, so we may assume $\delta^* < M$. This implies for our Linear Alignment instance that  
    \begin{align*}
        \forall {i \in [n]} ~ \exists {j \in[m]}:  \left| \left( A[i]+iM + sM +c \right) - (B[j] + jM) \right| \leq \delta^*,
    \end{align*}
    Clearly, $s \in \{0, ..., m-n\}$ as $i+s \leq m$ holds due to the large value of $M$.
    For the same reason, we can set $j = i+s$ and thereby
    \begin{align*}
        \forall {i \in [n]} :  \left| \left( A[i]+c \right) - (B[i+s]) \right| \leq \delta^*,
    \end{align*}
    which concludes this direction of the proof.
    
    For the other direction, let $s' \in \{0, ..., m-n\}, c' \in \mathbb{R} $ such that the right-hand side is minimized with value $\delta'$, so we have
    \begin{align*}
        \forall i \in [n]: \left| \left( A'[i]+c' \right) - B'[i+s'] \right| \leq \delta'.
    \end{align*}
    We choose $\tau:= s'M+ c$ in the directed HuT distance, which gives 
    \begin{align*}
        \forall {i \in [n]} ~ \exists {j \in[m]}: \left| \left( A[i]+iM + s'M +c \right) - (B[j] + jM) \right| \leq \delta'
    \end{align*}
    since we can choose $j=i+s'$ for a given $i$, concluding the proof.
\end{proof}

By leveraging the reduction from $L_\infty$-Necklace Alignment to Linear Alignment\footnote{Essentially, it suffices to duplicate the array $B$; see \cite{BremnerCDEHILPT14} for a correctness proof.}, we get the following lower bound for the balanced version of HuT, i.e, with $n=\Theta(m)$.
\begin{corollary}
    \label{cor:necklace-LB}
    If 1-dimensional balanced HuT can be solved in time $T(n)$, then the $L_{\infty}$-Necklace Alignment problem can be solved in time $O(T(n)).$
\end{corollary}
Furthermore, by leveraging the reduction from MaxConv LowerBound to $L_\infty$-Necklace Alignment from~\cite{cygan2019ProblemsEquivalentMin+Convolution}, we can also get a lower bound from MaxConv LowerBound for the 1-dimensional balanced HuT problem.

\begin{definition}[MaxConv LowerBound, MaxConvLB]
    Given integer arrays $A,B,C$ of length $n$\footnote{Other common definitions use for $C$ an array of size $2n$. For convenience, we stick to this definition.}. Determine whether $C[k] \leq \max_{i+j=k} (A[i]+B[j])$ holds for all $k \in \{2, ..., n\}$. 
\end{definition}

\begin{corollary}
    \label{cor:balHuTMaxConv}
    If 1-dimensional HuT can be solved in time $T(n)$, then MaxConv LowerBound can be solved in time $O(T(n) \log(n)).$
\end{corollary}
Interestingly, we can show the same lower bound from MaxConv LowerBound for discrete HuT. 
\begin{theorem}
    \label{thm:maxconv-LB}
    If 1-dimensional discrete HuT can be solved in time $T(n)$, then the MaxConv LowerBound problem can be solved in time $O(T(n)).$
\end{theorem}
\begin{proof}
    Given arrays $A,B,C$ of positive values as inputs from MaxConv LowerBound.
    We construct $T, P, Q$ as arrays for DiscHuT.
    By abuse of notation, we allow $Q$ to be indexed negatively. 
    We define
    \begin{align*}
        T &= \left[-(kM + C[k])\right]_{k \in [n]}, \\
        P &= \left[ iM + A[i] + M/4 \right]_{i \in [n]}, \\
        Q &= \left[ -zM + M/4 \right]_{z \in \{-n, ..., 0\}}\cup \left[ -(jM + B[j]) \right]_{j \in [n]}, \\
        \delta &= \frac{M}{4} - 1,
    \end{align*}
    where $M$ is a large number with $M > 4\cdot\max_{i,j,k} \left( A[i] + B[j] + C[k] \right)$. 
    
    Assume we have a NO-instance to MaxConv LowerBound, i.e., there exists a $k^*$ such that for all $i < k^*$ we have $A[i] + B[k^* - i] < C[k^*]$. For $i < k^*$, we get
    \begin{align*}
    |p_i + \tau_{k^*} - q_{k^* - i}| &= |iM + A[i] + M/4 -(k^*M + C[k^*]) + (k^* - i)M + B[k^*-i]| \\
                                        &= | A[i] + B[k^* - i] - C[k^*] + M/4 | \\
                                        &\leq \delta.
    \end{align*}
    For $n \geq i \geq k^*$ we have $k^* - i \leq 0$ and thus:
    \begin{align*}
        |p_i + \tau_{k^*} - q_{k^* - i}| &= |iM + A[i] + M/4 -(k^*M + C[k^*]) - ((- (k^* - i))M + M/4) |\\
        &=| A[i] - C[k^*] | \\ 
        &\leq \delta.
    \end{align*}
    Thus in both cases, we have a YES-instance to DiscHuT.
    
    If otherwise we have a YES-instance to MaxConv LowerBound, for all $k$ there exists a $i < k$ such that $A[i] + B[k - i] \geq C[k]$. Let $\tau_k, p_i$ be elements of the Hausdorff instance corresponding to such a $(C[k], A[i])$ witness. We have a NO-instance to DiscHuT since
    \begin{align*}
        |p_i + \tau_{k} - q_{k - i}| &= |iM + A[i] + M/4 - (kM + C[k]) + (k - i)M + B[k-i] | \\
        &=| A[i] + B[k - i] - C[k^*] + M/4 | \\
        &> \delta 
    \end{align*}
    and taking any other point $q' \neq q_{k-i}$ of $Q$ results in a distance of at least $M/2$.
\end{proof}

\subsection{Relation to \texorpdfstring{\FOPZ}{FOP\_Z}~Problems}
\label{sec:apx-fopz-problems}
We give an upper bound to discrete HuT, which stems from a reduction from 3-dimensional Discrete HuT  to 3-SUM. The basic idea consists of a complement trick, which enables a quantifier switch in the discrete HuT problem.

\begin{theorem}
    \label{thm:all-ints-3sum-easiness}
    Let $d \leq 3$.
    If All-Ints-$3$-SUM$(n,m,k)$ can be solved in time $\mathcal{T}(n,m,k)$, then we can solve $d$-dimensional discrete HuT, in time $\Tilde{O}(\mathcal{T}(|T|,|P|,|Q| \log |Q| )+ |Q| \log |Q|+ \max(|P|,|T|)).$
\end{theorem}
\begin{proof}
    Let $T, P, \mathcal Q, \delta$ be the input to the Translation Problem with Hypercubes, which is equivalent to discrete HuT.
    Be reminded that we want to satisfy the formula
    \begin{align}
        \exists \tau \in T ~ \forall p \in P ~ \exists \hat Q \in \mathcal Q: (p+\tau) \in \hat Q. \label{eq:hypercube_HuT}
    \end{align} 
    Let $\bar {\mathcal Q} = \overline{ \bigcup \mathcal Q }$ be the complement of the hypercubes spanned by the points in $Q$. 
    The cubes in $\mathcal Q$ are congruent, so we can make use of the result of Agarwal, Sharir and Steiger~\cite{agarwal_decomposing_2024}, to decompose $\bar {\mathcal Q}$ into $O(|\mathcal Q| \log |\mathcal Q|)$ disjoint boxes in runtime $O(|\mathcal Q| \log(|\mathcal Q|)).$
    Let $\bar q_i^+$ be the upper bound of the interval of the box $\bar {Q} \in \bar{\mathcal Q}$ in dimension $i$, and $\bar q_i^-$ the lower bound. In other words, the box $ \bar{Q}$ is of the form $[\bar{q_1}^-, \bar{q_1}^+] \times [\bar{q_2}^-,\bar{q_2}^+] \times [\bar{q_3}^-,\bar{q_3}^+].$ 
    Notice that for fixed $(p,\tau) \in P \times T$, it holds that $\lnot \exists \hat Q \in \mathcal Q: p+t \in \hat Q \iff \exists \bar{Q} \in \bar{\mathcal Q}: p+t \in \bar{Q}.$
    Thus, by a simple negation, we get a NO-instance to discrete HuT if and only if $\forall \tau \in T ~ \exists p \in P ~ \exists \bar Q \in \bar{\mathcal Q}: \tau + p \in \bar q$ holds. 
    
    We rewrite this formula, by making use of the coordinates spanning each box.
    Thus, equivalently to deciding \Cref{eq:hypercube_HuT}, it suffices to decide a $2d$-dimensional $\forall\exists\exists$-formula  of the form \[ 
        \forall \tau \in T ~ \exists p \in P ~ \exists \bar Q \in \bar{\mathcal Q}: \bigwedge_{i \in [d]} (\tau_i + p_i \leq \bar q_i^+) \land (-\tau_i - p_i \leq -\bar q_i^-).
    \] 
    This formula can be reduced to $\Tilde{O}(1)$ calls to All-ints-$3$-SUM$(|T|,|P|,|\mathcal Q| \log |\mathcal Q|)$ by \cite{GokajKST25}.
\end{proof}
We remark, that if the sets are of the same size, we can leverage the equivalence between $3$-SUM and All-Ints $3$-SUM \cite{williams2018SubcubicEquivalencesPath}, to get the following: 
\begin{corollary}
    \label{cor:3sum-easiness}
    Let $d \leq 3$. If $3$-SUM can be solved in time $T(n)$, then we can solve $d$-dimensional balanced Discrete HuT in time $\Tilde{O}(T(n\log (n))+n \log (n)).$
\end{corollary}

Interestingly, we can also show hardness of discrete HuT under the class $\FOPZ_{,d\leq 3}(\forall\exists\exists)$.
We first state the result for varying sizes of the sets $A,B$ and $C$, and afterwards for the balanced version, i.e., $|A|=|B|=|C|$.
\begin{theorem}
    \label{thm:fopz-aee-hardness}
    Let $d\leq 3$. If $2d$-dimensional Discrete HuT can be solved in time $T_{2d}(n,m,k)$, then we can solve all $d$-dimensional inequality formulas of the class $\mathsf{FOP}_\mathbb{Z}  (\forall \exists \exists)$ with inputs $A,B,C$ in time $O(T_{2d}(|A|,|B|,|C| \log |C|) + |C| \log |C| +\max(|A|,|B|)).$ 
\end{theorem}

\begin{proof}
    We proceed analogously to the proof of \cite{gokaj_completeness_2025}[Theorem 28], diverging from it at the final steps, by adding the complementation trick, which was also used in the proof of \Cref{cor:3sum-easiness}.
    
    Assume we are given a $\mathsf{FOP}_{\mathbb{Z}}(\forall \exists \exists)$ formula of inequality dimension at most 3 which is $\phi:= \forall a \in A \exists b \in B \exists c \in C: \varphi$ in DNF, where $\varphi$ is a quantifier free linear arithmetic formula. Let the atoms of $\varphi$, be $L_1,L_2,L_3$, where $L_i$ is of the form $\alpha_i^Ta + \beta_i ^Tb \leq \gamma_i^Tc  + S_i$, after replacement of the free variables.
    Construct the sets :
     \begin{align*}
   A':= \left \{ \left( \begin{array}{cc}
    \alpha_{1}^T a \\
    \alpha_{2}^T a \\
    \alpha_{3}^T a
    \end{array}  \right): a \in A \right \},
    B':=\left \{ \left( \begin{array}{cc}
      \beta_{1}^T b \\
      \beta_{2}^T b \\
      \beta_{3}^T b 
      \end{array}  \right):b \in B \right \},
      C':=\left \{ \left( \begin{array}{cc}
        \gamma_{1}^T c +S_{1} \\
        \gamma_{2}^T c +S_{2}\\
        \gamma_{3}^T c+ S_{3}
        \end{array}  \right):c \in C \right \}
  \end{align*}
  Let the co-clauses of $\varphi$ be $V_1, \dots ,V_h$. Thus, we aim to decide a formula of the form:
  \begin{equation}\label{eq:Dnfcurr}
  \forall a' \in A' \exists b' \in B' \exists c' \in C': \bigvee_{i=1}^{h} V_{i} ,
  \end{equation} 
 where each co-clause $V_i$  is of the form :
  $ \bigwedge_{k \in V_i^K} L_k  \land \bigwedge_{j \in V_i^J} \lnot L_j,$ 
  for some $V_i^J,V_i^K \subseteq \{1,2,3\}$ and $V_i^J\cap V_i^K =\emptyset$. 
  For a $c' \in C'$ and a co-clause $V_i$, define a solution space for the co-clause as: 
  $$\mathcal{S}(V_{i},c'):=\{x \in \mathbb{R}^3 : \bigwedge_{k \in V_i^K}x[k] \leq c'[k] \land \bigwedge_{j \in V_i^J}x[j] \geq c'[j]+1   \}.$$
  Clearly, $\mathcal{S}(V_{i},c')$ is an orthant in $\mathbb{R}^3$, and for fixed $(a',b') \in A' \times B' $ a co-clause $V_i$ is fulfilled by the assignment $(a',b',c')$ iff $a'+b' \in \mathcal{S}(V_{i},c')$.

  By a standard argument\footnote{Essentially, it suffices to cap the size of the orthants after an appropriately large range.}(see \cite{gokaj_completeness_2025} for precise details), we transform the orthants $\mathcal{S}(V_i,c')$ into congruent cubes $\mathcal{C}_{c'}^{i}$, such that for fixed $(a',b') \in A' \times B'$ it holds that $a'+b' \in \mathcal{S}(V_{i},c') \iff a'+b' \in \mathcal{C}_{c'}^{i} $.
  So, it suffices to check the satisfiability of the formula:
    \begin{align}
      &\forall a' \in A' \quad \exists b' \in B' \quad \exists c' \in C': \bigvee_{i=1}^h a'+b' \in \mathcal{C}_{c'}^i \\ \iff  &\forall a' \in A' \quad \exists b' \in B':  a'+b' \in \bigcup_{c' \in C', i \in [h]}\mathcal{C}_{c'}^i 
     \\ \iff   &\lnot \left( \exists a' \in A' \quad \forall b' \in B' : (a'+b') \not \in \bigcup_{c' \in C', i \in [h]}\mathcal{C}_{c'}^i \right)\label{eq: good}
    \end{align}
    We wish to finally reduce $\left( \exists a' \in A' \quad \forall b' \in B' : (a'+b') \not \in \bigcup_{c' \in C', i \in [h]}\mathcal{C}_{c'}^i \right)$ to an instance of discrete HuT. 
     We decompose the complement of the inner part $\bigcup_{c' \in C', i \in [h]}\mathcal{C}_{c'}^i$ into a set of $O(|C| \log |C|)$ disjoint boxes $\mathcal{R}$ in time $O(|C|\log |C| )$ by the algorithm of \cite{agarwal_decomposing_2024}. 
     For fixed $(a',b') \in A' \times B'$ it holds that 
     \begin{align*}
      \left ( (a'+b') \not \in \bigcup_{c' \in C', i \in [h]}\mathcal{C}_{c'}^i\right) \iff \exists R \in \mathcal{R}: a'+b' \in R.
     \end{align*}
         Let $ r_i^+$ be the upper bound of the interval of the box $R$ in dimension $i$, and $r_i^-$ the lower bound. In other words, the box $ R \in \bar{R}$ is of the form $[r_1^-, r_1^+] \times [r_2^-,r_2^+] \times [r_3^-,r_3^+].$ 
    Finally, we have reduced Equation \eqref{eq: good} to 
    \begin{align}
        \exists a' \in A' \quad \forall b ' \in B' \quad \exists R \in \mathcal{R}: \bigwedge_{i\in[d]}  (a'_i+ b'_i \leq r_i^+) \land (-a'_i - b'_i \leq -r_i^-). \label{eq:ortant-translation}
    \end{align}
    By \cite{gokaj_completeness_2025}, the above formula can be reduced to an instance of $2d$-dimensional discrete HuT, in linear time, with $|T|=|A|, |P|=|B|$ and $|Q|= |C|\log |C|$.
\end{proof}
We remark that for the balanced case, we get the following lower bound.
\begin{corollary}
    Let $d \leq 3$. If $2d$-dimensional balanced Discrete HuT can be solved in time  $T_{2d}(n)$, then we can solve all $d$-dimensional inequality formulas of the class $\mathsf{FOP}_{\mathbb{Z}}(\forall \exists \exists), $  of input size $n'$, in time $O(T_{2d}(n' \log n') + n' \log n').$
\end{corollary}

\section{Transferring Lower Bounds to the Undirected Setting}
\label{sec:apx-undirected}

As of now, all here presented lower bounds only apply to the directed Hausdorff under Translation setting. 
Chan \cite{Chan23} already raised the question if these bounds transfer to the undirected setting, however there are no such lower bounds known so far for $d \geq 3$. 
Only for $d = 2$, lower bounds for the undirected setting are known \cite{BringmannN21}, which match the bounds for the directed variant. 
For $d=1$, we have an optimal $\tOh(n + m)$ algorithm \cite{rote1991ComputingMinimumHausdorff}.

We resolve this question by providing a simple observation that allows a reduction from the Translation Problem with Orthants to undirected Hausdorff under Translation, maintaining all instance parameters. 
Using this reduction, all lower bounds for HuT in $d \geq 3$ dimensions presented in this paper can be directly applied to undirected Hausdorff under Translation.
More generally, we can reduce any instance of directed Hausdorff under Translation to the undirected setting at the cost of doubling the dimension using \Cref{lem:box-to-cube}.

We start by giving our observation that allows the promised reduction.
Let $P, Q$ be the point sets of a HuT instance with $\delta >0$, and assume that the set of admissible translations $T \subseteq [-\delta, \delta]^d$ is restricted. 
We note that $\max_{\tau \in T} |(q + \tau) - q|_\infty \leq \delta$, so using $P' = P \cup Q$ and $Q' = Q$ as an instance for the undirected variant ensures that each element $q \in Q$ is in the $\delta$-vicinity of its respective point $q + \tau = q' \in P' + \tau$ for arbitrary $\tau \in T$. The analogue holds for the points $Q + \tau \subset P' + \tau$.  
Thus, the directed problem reduces to the undirected variant as a feasible translation $\tau$ to the undirected setting only needs to fulfill $\hddist(P + \tau, Q) \leq \delta$.

\begin{observation}
    \label{obs:directed-to-undirected-base}
    If we restrict the set of admissible translations $T \subseteq [-\delta, \delta]^d \subset \mathbb{R}^d$, for all $\tau \in T$ it holds that \[
       \hddist(P+\tau, Q) \leq \delta \iff  \hddistun((P \cup Q)+\tau, Q) \leq \delta.
    \]
\end{observation}

We turn this observation into a reduction from the Translation Problem with Orthants.
Unfortunately, it introduces an extra $|Q|$ summand to the input size, which is undesirable in a lopsided setting.
However, as the set of translated points $P + T$ is restricted in a TPwO instance, we are able to omit the extra $|Q|$ points in the resulting instance by introducing a few, strategically placed points. 

\begin{figure}
    \centering
    \includegraphics[width=0.5\linewidth]{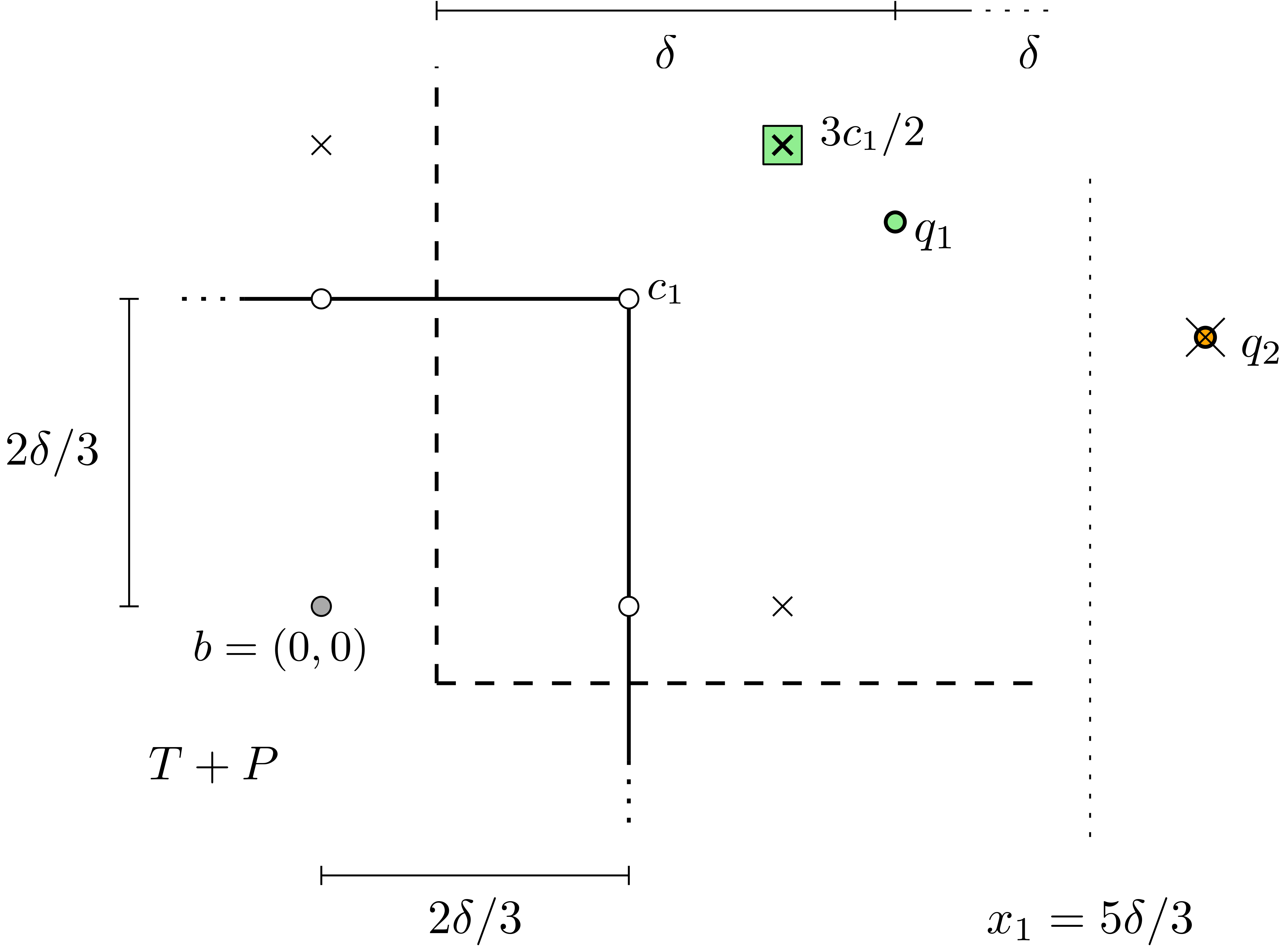}
    \caption{An example of the reduction in \Cref{thm:reduction-orthants-to-undirected} in $2$-D. 
    The region $T + P$, centered around a relative origin $b$, is bounded by $[-2\delta/3, 2\delta/3]^2$. 
    The points $q_1, q_2$ are elements of $Q$, while the white points, especially $c_1$, are elements of the constructed set $C$.
    The point $q_2 \in Q$ will be removed since $(q_2 + [-\delta, \delta]^2) \cap [-2\delta/3, 2\delta/3]^2 = \emptyset$. The closest point to $q_1$ is $c_1$, so $3/2~c_1$ is added to the resulting point set $P'$. }
    \label{fig:directed-to-undirected}
\end{figure}

\begin{theorem}
    \label{thm:reduction-orthants-to-undirected}
    If we can decide the undirected Hausdorff under Translation problem in $d$-D with point sets of size $n, m$ in time $T(n, m)$, we can decide the Translation Problem with Orthants in $d$-D with $n$ points and $m$ orthants in time $T(n, m) + O(n+m)$. 
\end{theorem}

\begin{proof}
    Recall from \Cref{sec:preliminaries} that a TPwO instance consists of a set of TPwC instances as well as a global target box $B$, with maximum side length $\delta_0$, bounding the translated points.
    Given an instance of the Translation Problem with Orthants, we consider its TPwC sub-instances individually, i.e., we give a reduction from each sub-instance to a uHuT instance.
    Additionally, the target box $B$ implies a bounding box with side length less than $\delta_0$ for the set of possible translations $T$, we thus introduce two more uHuT instances of size $2$ that allow exactly the translations in $T$, analog to \Cref{lem:tpwo-to-tpwc}.
    Thus, we are left with a set of uHuT instances $(P_i, Q_i)_i$ and search for a translation feasible for all these instances. We promise that all instances have a bounding box that does not exceed $4\delta_0$. 
    Again like in \Cref{lem:tpwo-to-tpwc}, we translate each instance by a vector $(i9\delta_0, 0, ..., 0)$ so that using a feasible translation all points in $P_i$ are sufficiently close only to points in $Q_i$, and vice versa. 

    It remains to show that we can reduce each TPwC sub-instance to an uHuT instance. 
    For an individual TPwC instance, given the point set $P$, $|P| = n'$, and the set of hypercubes $\mathcal C$, $|\mathcal C| = m'$, of side length $\delta$. 
    Note that the target box $B$ implies that the translated point set $P + T$ and the set of feasible translations $T$ is restricted. 
    By the definition of TPwO, we may assume that $\delta \geq \delta_0$ such that $T \subseteq [-\delta/3, \delta/3]^d$ and $P + T \subseteq B \subseteq b + [-2\delta/3, 2\delta/3]^d$, where $b \in \R^d$ is some relative origin.
    Let $V(\mathcal C)$ be the vertices of the hypercubes inside the bounding box.
    We next strategically place points in the TPwC sub-instance to allow a reduction to uHuT.

    We first simplify the instance.
    We remove any hypercube with center $q \in Q$ for which $q \notin b + [-5\delta/3, 5\delta/3]^d$ because no point of $T + P$ can lie within the hypercube of $q$.
    If we have a center point $q \in b + [-\delta/3, \delta/3]^d$, we replace the instance with a single point $p$ and hypercube with center point $p$, a trivial $\YES$-instance to uHuT.

    We now construct the undirected Hausdorff under Translation instance.
    Let $C' = \{-2\delta/3, 0, 2\delta/3\}^d$ be the set of middle points and vertices of all faces of the bounding box of $P + T$.
    The set $C = \{b + 3c/2 \mid \exists q \in Q: c = \argmin_{c' \in C'} (|q - (c' + b)|_\infty) \}$ are the scaled points of $C'$ that are closest to a point in $Q$.
    The sets $C', C$ have size $O(3^d) = O(1)$ and can be trivially constructed in time $O(m)$. 
    For our undirected Hausdorff under Translation instance, we set $P' = P \cup C$ and $Q' = Q$ and the distance value $\delta$.

    For correctness, we show that for any translation $\tau \in T$, each point in $C$ is within distance $\delta$ to a point in $Q$ and vice versa.
    As a result, we get \[
        \hddist(P+\tau, Q) \leq \delta \iff \hddistun((P \cup C)+\tau, Q) \leq \delta,
    \]
    which shows equivalence of the TPwO instance, written here in its directed HuT formulation, to the constructed undirected HuT instance.
    
    Given a point $c \in C$. As $c = b + 3/2 \cdot \argmin_{c' \in C'} (|q - (c' + b)|_\infty)$ for some center $q \in Q$, we know that $q \in c + [-2\delta/3, 2\delta/3]^d$ since this is a superset of the space in $b + [-5\delta/3, 5\delta/3]^d$ that is closest to $c$ from all $C$.
    As $T \subseteq [-\delta/3, \delta/3]^d$, we have $q + T \in c + [-\delta, \delta]^d$, so $c$ lies in $\delta$-distance to $q$ under all translations $T$.

    As $C \subseteq b + \{-\delta, 0, \delta\}^d$ and $Q \subset b + [-5\delta/3, 5\delta/3]^d$ by our construction and reduction, each point in $q \in Q$ has at least one point $c \in C$ for which $q \in c + [-2\delta/3, 2\delta/3]^d$ holds. Again we have $q + T \in c + [-\delta, \delta]^d$, and $q$ lies in $\delta$-distance to $c$ under all translations $T$.
\end{proof}

This result already allows us to transfer the majority of lower bounds presented in this paper to the undirected Hausdorff under Translation Problem, as they use the Translation Problem with Orthants as an intermediate problem.
However, we can also generally reduce the directed version to the undirected setting by using \Cref{lem:box-to-cube}. Note that HuT has an equivalent formulation as a Translation Problem with Boxes because hypercubes are boxes, see \Cref{sec:preliminaries}. We can reduce it to the Translation Problem with Orthants at the expense of doubling the dimension. 

\begin{corollary}
    \label{cor:directed-to-undirected-reduction}
    If undirected continuous Hausdorff under Translation in $2d$-D can be solved in time $T_{2d}(n,m)$, then directed continuous Hausdorff under Translation in $d$-D can be solved in time $O(T_{2d}(n,m))$. 
    
    Likewise, if undirected discrete Hausdorff under Translation in $2d$-D can be solved in time $T_{2d}(n,m,t)$, then directed discrete Hausdorff under Translation in $d$-D can be solved in time $O(T_{2d}(n,m,t))$. 
\end{corollary}

For a better overview which of our lower bound results apply (in which number of dimensions) for the undirected setting, we here list the most important results, all following from either \Cref{thm:reduction-orthants-to-undirected} or \Cref{cor:directed-to-undirected-reduction}.
As the roles of $P$ and $Q$ are interchangeable in the undirected setting, we can omit restrictions of the form $n \geq m$ for all lower bounds. 

We start with the lower bounds from \Cref{sec:apx-higher-dimension-overall}.
While for \Cref{thm:lopsided-directed-LB,thm:pcd-reduction-3D} we explicitly construct a TPwO instance, the lower bound of Chan \cite{Chan23} also reduces to the Translation Problem with Orthants, we refer the reader to the definition of TPwO in \Cref{sec:translation-problems} as well as \cite[Problem 4, Lemma 3]{Chan23}.

\begin{corollary}
    \label{cor:lower-bounds-undirected}
    Let $d \geq 2, \varepsilon > 0$, and $\lambda \in (0, 1]$. For $\min(n,m) \in \Theta(\max(n,m)^\lambda)$, we get the following lower bounds for undirected Hausdorff under Translation in $d$-D.
    \begin{enumerate}
        \item 
            There is no $(nm)^{d/2 - \varepsilon}$ combinatorial algorithm under the $k$-clique hypothesis. 
        \item 
            There is no $(nm)^{\frac{d\omega}{6} - \varepsilon}$ (non-combinatorial) algorithm under the $k$-clique hypothesis.
        \item 
            There is no $(nm)^{\lfloor \frac{d}{2} \rfloor \frac{\lambda + 3}{3 \lambda + 3} - \varepsilon}$ (non-combinatorial) algorithm under the $k$-hyperclique hypothesis.
    \end{enumerate}
    Moreover, for $d=3, \lambda \in (0, 1/2]$ and $\min(n,m) \in \Theta(\max(n,m)^{\lambda})$, there is no $O((nm)^{3/2 - \varepsilon})$ (non-combinatorial) algorithm for undirected Hausdorff under Translation under the $k$-hyperclique hypothesis.
\end{corollary}

Moreover, we can transfer lower bounds from \Cref{sec:apx-additive-problems}. 
As the reduction from the class of $\FOPZ(\forall\exists\exists)$ to the discrete Hausdorff under Translation variant (\Cref{thm:fopz-aee-hardness}) reduces to the Translation Problem with Orthants (see Eq.~\ref{eq:ortant-translation}), we get a directed reduction to the undirected variant. 
For the lower bounds of \Cref{sec:lower-bounds-1D}, this is not the case so we have to double the dimension.

\begin{corollary}
    \label{cor:fopz-aee-hardness-undirected}
    Let $d\leq 3$. If $2d$-dimensional undirected discrete HuT can be solved in time $T_{2d}(n,m,k)$, then we can solve all $d$-dimensional inequality formulas of the class $\mathsf{FOP}_\mathbb{Z}(\forall \exists \exists)$ with inputs $A,B,C$ in time $O(T_{2d}(|A|,|B|,|C| \log |C|) + |C| \log |C| + \max(|A|,|B|)).$ 
\end{corollary}

\begin{corollary}
    \label{cor:maxconvLB-LB-undirected}
    If 2-dimensional undirected discrete HuT can be solved in time $T_2(n)$, then the MaxConv LowerBound problem can be solved in time $O(T_2(n) + n)$.
\end{corollary}

At last, we come back to the folklore reduction of \Cref{obs:undir-to-dir}, reducing an instance of undirected HuT to an instance of directed HuT.
Unfortunately for lopsided input sets, this reduction introduces additional points to both input sets.
By the here proven lower bounds for undirected Hausdorff under Translation, we can actually show that a reduction to a single instance of HuT maintaining the input sizes is impossible to achieve.
For a small constant $\lambda \in (0, 1/9]$ with $|P| \in O(|Q|^\lambda)$, \Cref{cor:lower-bounds-undirected} gives a $(|P||Q|)^{d/2 - o(1)}$ lower bound for uHuT. 
However, \Cref{thm:lopsided-algo-3D} gives for the same sizes of $P, Q$ a faster than $\tOh(|P||Q|)^{3/2}$ algorithm in $3$-D for directed HuT.
Thus, a reduction from the undirected setting to a single instance of the directed setting maintaining the input sizes would violate our proven lower bounds.

\section{Open Questions}

\begin{itemize}
    \item For the balanced case, i.e., $m = \Theta(n)$, our (non-combinatorial) lower bounds do not match the combinatorial bound of $(nm)^{d/2}$.
    On the $n\geq m$ side, we were not able to extend our matching bounds (\Cref{thm:pcd-reduction-3D}) to $n^{1/2} < m \leq n$.
    On the $m \geq n$ side, the generalization in \Cref{thm:lopsided-directed-LB} is lossy so we could not give matching bounds here either.
    For finding a non-combinatorial tight lower bound for $n=\Theta(m)$, one must take into account that any $(nm)^{d/2-o(1)}$ lower bound cannot apply for $m \gg n$ (\Cref{cor:d/2-floor-algo}).
    It seems reasonable to us that a two sided approach, having a lower bound for $n \geq m$ and $m \geq n$ separately, may yield the required properties.
    Looking at the necessary approach, we ask: Can we extend our approach of \Cref{thm:pcd-reduction-3D}, can we find a lossless generalization of \Cref{thm:lopsided-directed-LB}, or do we need to develop a completely different tool?

    \item For $d=3$, there exists a combinatorial lower bound of $(nm)^{d/2-o(1)}$ for $n\geq m$ \cite{Chan23} and a $(nm)^{\lfloor d/2 \rfloor + \varepsilon}$ algorithm for $m \gg n$ and $\varepsilon > 0$ (\Cref{cor:d/2-floor-algo}).  
    As raised at the end of \Cref{sec:apx-other-side}, we ask where the cutoff point lies, i.e., for which values of $n,m$ do we start to get a faster than $(nm)^{d/2 - o(1)}$ algorithms? 
    
    \item We reduce Discrete Hausdorff under Translation to All-Ints 3SUM and 3SUM for inequality dimension $d\leq 3$ (\Cref{thm:fopz-aee-hardness,cor:3sum-easiness}), partially answering an open question of \cite{gokaj_completeness_2025}. 
    However, we fail to do so for larger dimension. 
    Our current techniques rely on the decomposition (the complement of) a union of $n$ cubes into $\tOh(n)$ disjoint boxes, which is not possible for $d \geq 4$, so the approach seems to be inherently limited to $d \leq 3$ dimensions. 
\end{itemize}

%%%%%%%%%%%%%%%%%%%%%%%%%%%%%%%%%%%%%

\bibliographystyle{plainurl}
\bibliography{hausdorff-ref}

\appendix
\section{Relevant Problem Definitions and Hardness Hypotheses}
\label{sec:definitions}

\subsection{Problems}
\begin{definition}[$L_\infty$-Necklace Alignment]
    Given arrays $A$,$B$ of size $n$ with entries in $[0,1).$ 
    Compute
    \[ \min_{c \in [0,1), s \in [n]}\max_{i=1}^{n} \{ d((A[i]+c) \bmod 1, B[i+s \bmod n]) \},\]
    where $d(a,b)$ denotes the minimal between the clockwise and counter-clockwise  distances along the unit-circle perimeter, i.e  for $a,b \in [0,1($ the value  $d(a,b)$ is defined as
    $d(a,b):=\min \left(|a-b|, 1-|a-b| \right).$
\end{definition}

\begin{definition}[Linear Alignment problem]
    Given sorted vectors $A,B$ of real numbers of size $n$ and $m$ respectively, where $m \geq n.$ 
    Compute an integer $s< m-n$ and a real-valued $c \in [0,1)$, which minimize the value 
    \[ \max_{i\in[n]} \left(\left | \left(A[i] +c \right) - B[i+s]  \right | \right).\]
\end{definition}

\begin{definition}[MaxConv LowerBound]\label{def:max_convlb}
    Given integer arrays $A,B,C$ of length $n$. Determine whether 
    $C[k] \leq \max_{i+j=k} (A[i]+B[j])$ holds.
\end{definition}

\begin{definition}[3-SUM]
    Given sets $A,B,C$ of integers of size at most $n$. Do there exist $a \in A, b \in B, c \in C$ such that $a+b+c=0.$
\end{definition}

\begin{definition}[All-Ints $3$-SUM]
    Given sets $A,B,C$ of integers of size at most $n$. Determine for each $a \in A$, whether there exist $b \in B, c \in C$ such that $a+b+c=0$.
\end{definition}
\begin{definition}[All-Ints $3$-SUM $(n,m,k)$]
    Given sets $A,B,C$ of size $n,m,k$ respectively. Determine for each $a \in A$, whether there exist $b \in B, c \in C$ such that $a+b+c=0$.
\end{definition}

\begin{definition}[$\mathsf{FOP}_{\mathbb{Z}}(\exists \forall \exists)$]
    The class of problems $\mathsf{FOP}_{\mathbb{Z}} (\exists \forall \exists)$ consists of all problems whose inputs are sets $A_1 \subseteq \mathbb{Z}^{d_1}, A_2 \subseteq \mathbb{Z}^{d_2}, A_3 \subseteq \mathbb{Z}^{d_3}$, and free variables $t_1 \dots t_\ell \in \mathbb{Z}$, and the task is to decide whether $\exists a_1 \in A_1 \forall a_2 \in A_2 \exists a_3 \in A_3$ such that \[\varphi(a_1[1], \dots, a_1[d_1], \dots, a_k[1], \dots, a_k[d_k], t_1 \dots, t_\ell)\] holds,
    where $\varphi$ is a fixed, but arbitrary linear arithmetic formula over the vector entries of $a_1, a_2, a_3$ and the free variables.  The dimensions $d_1,\dots, d_k$ are constant.
\end{definition}

\begin{definition}[$\mathsf{FOP}_{\mathbb{Z}}(\forall \exists \exists)$]
    The class of problems $\mathsf{FOP}_{\mathbb{Z}} (\forall \exists \exists)$ consists of all problems whose inputs are sets $A_1 \subseteq \mathbb{Z}^{d_1}, A_2 \subseteq \mathbb{Z}^{d_2}, A_3 \subseteq \mathbb{Z}^{d_3}$, and free variables $t_1 \dots t_\ell \in \mathbb{Z}$, and the task is to decide whether $\forall a_1 \in A_1 \exists a_2 \in A_2 \exists a_3 \in A_3$ such that \[\varphi(a_1[1], \dots, a_1[d_1], \dots, a_k[1], \dots, a_k[d_k], t_1 \dots, t_\ell)\] holds,
    where $\varphi$ is a fixed, but arbitrary linear arithmetic formula over the vector entries of $a_1, a_2, a_3$ and the free variables.  The dimensions $d_1,\dots, d_k$ are constant.
\end{definition}

\begin{definition}[Inequality Dimension of a Formula]
    Let $\phi = Q_1 x_1\in A_1, \dots, Q_k x_k\in A_k: \psi$ be a $\mathsf{FOP}_{\mathbb{Z}}$ formula with $A_i\subseteq \mathbb{Z}^{d_i}$.
    The \emph{inequality dimension} of $\phi$ is the smallest number $s$ such that there exists a Boolean function $\psi' :\{0,1\}^s \to \{0,1\}$ and (strict or non-strict)
    linear inequalities $L_1, \dots, L_s$ in the variables $\{x_i[j] : i\in \{1,\dots,k\} ,j\in \{1,\dots,d_i\} \}$ and the free variables such that $\psi(x_1,\dots, x_k)$ is equivalent to $\psi'(L_1,\dots, L_s)$.  
\end{definition}

\begin{definition}[The $k$-Clique Problem]
\label{def:k-clique}
    Given a graph, is there a $k$ element subset of vertices $C$ such that each pair of vertices in $C$ forms an edge?
\end{definition}

\begin{definition}[The Colorful $k$-Clique Problem]
\label{def:colorful-k-clique}
    Given a $k$-partite graph with vertex set $V = V_1 \dot\cup ..., \dot\cup V_k$, is there a $k$-tuple $C \in V_1 \times ... \times V_k$ such that each pair of vertices in $C$ forms an edge?
\end{definition}

\begin{definition}[The $k$-Hyperclique Problem]
\label{def:k-hyperclique}
    Given a $3$-uniform hypergraph, is there a $k$ subset of vertices $C$ such that each $3$-subset of $C$ forms a hyper-edge?
\end{definition}
    
\begin{definition}[The Colorful $k$-Hyperclique Problem]
\label{def:colorful-k-hyperclique}
    Given a $3$-uniform $k$-partite hypergraph with vertex set $V = V_1 \dot\cup ..., \dot\cup V_k$, is there a $k$-tuple $C \in V_1 \times ... \times V_k$ such that each $3$-subset of $C$ forms a hyper-edge?
\end{definition}

\subsection{Hardness Hypotheses}

We use the following hardness assumptions from fine-grained complexity theory (see~\cite{williams2018some} for a general overview): (1) the $k$-clique hypothesis, see~\cite{abboud_if_2018} and references thereof; (2) the 3-uniform hyperclique hypothesis, see~\cite{lincoln_tight_2018} and references thereof and (3) the 3SUM Hypothesis, see~\cite{gajentaan_class_1995} and references thereof.

\begin{definition}[The (Combinatorial) $k$-Clique Hypothesis]
\label{hyp:k-clique}
    Let $\varepsilon >0$.
    For $k \geq 3$ the problem of finding a $k$-clique (\Cref{def:k-clique}) in a graph of $n$ vertices cannot be solved by a combinatorial algorithm in $O(n^{k-\varepsilon})$ run-time.
    Equally, the problem of finding a colorful $k$-clique (see \Cref{def:colorful-k-clique}) in a graph of $n$ vertices cannot be solved by a combinatorial algorithm in $O(n^{k-\varepsilon})$ run-time.
    For non-combinatorial algorithms, it is not possible to find a (colorful) $k$-clique in a graph of $n$ vertices in time $O(n^{\omega k/3 - \varepsilon})$. 
\end{definition}

\begin{definition}[The $k$-Hyperclique Hypothesis]
\label{hyp:k-hyperclique}
    Let $\varepsilon >0$.
    For $k > 3$ the problem of finding a $k$-hyperclique in a $3$-uniform graph (\Cref{def:k-hyperclique}) of $n$ vertices cannot be solved in $O(n^{k-\varepsilon})$ run-time.
    Equally, the problem of finding a colorful $k$-hyperclique in a $3$-uniform graph (see \Cref{def:colorful-k-hyperclique}) of $n$ vertices cannot be solved in $O(n^{k-\varepsilon})$ run-time.
\end{definition}

\begin{definition}[The 3-SUM Hypothesis]
\label{hyp:3sum}
    Let $\epsilon>0$. There is no algorithm solving $3$-SUM in time $O(n^{2-\epsilon})$.
\end{definition}
We work with the following non-standard hypothesis, which has also been used in \cite{gokaj_completeness_2025}, concerning the MaxConv LowerBound problem; see also Definition \ref{def:max_convlb}.
\begin{definition}[MaxConv LowerBound hypothesis]
    Let $\epsilon>0$. There is no algorithm solving MaxConv LowerBound in time $O(n^{2-\epsilon})$.
\end{definition}
\section{Baseline algorithms}
\label{sec:baseline}
We shortly describe some algorithms for Hausdorff distance under translation.

\paragraph*{Almost-Linear-Time Algorithm for Discrete HuT in $1$-D}
\label{sec:discHuT-1D-linear}
The algorithm by Rote \cite{rote1991ComputingMinimumHausdorff} for the continuous case constructs the Hausdorff distance under Translation of single points $p \in P$ to the full set $Q$, which is a piecewise linear function of the used translation. The Hausdorff distance under Translation of the full set $P$ then is the envelope function of its parts. 
He shows that the full envelope function only has $O(n+m)$ many extremal points and gives an $\tilde O(n+m)$ algorithm to construct it.
When given a sorted set of candidate translations $T$, as in the discrete Hausdorff under translation variant, we can linearly scan the envelope function in $O(n+m)$ for the positions of the translations, and note their respective Hausdorff distance. 

\paragraph*{Algorithms for Discrete HuT}
\label{sec:discHuT-baseline-algorithms}
As before, we denote $t = |T|, n = |P|, m = |Q|$ for $T, P, Q$ being the input for discrete Hausdorff under Translation. Folklore claims an algorithm of runtime $O\left(t n \log^d m) +  m\log^d(m) \right).$

The idea is to construct range trees to store the set $\{q + [-\delta, \delta]^d \mid q \in Q \}$. Afterwards, we query for all combinations of $\tau \in T, p \in P$ whether there exists an element at position $\tau + p$. This gives a runtime of $O((tn + m)\log^d m)$.

For the undirected variant, we run the algorithm a second time for switched values for $P, Q$, storing which translation satisfies both runs. The runtime for that is still $\tOh(tn + m)$.

Alternatively, any algorithm explicitly constructing the region $X$ of feasible translations can be used for the discrete setting as one can additionally check for each of the $t$ candidate translation whether it lies inside $X$. As the feasible region consists of axis-aligned boxes, one can store the constructed feasible regions within a range tree in $O(|X| \log^d |X|)$ time. The checks for a feasible candidate translation can be performed in $O(t \log^d |X|)$.

\paragraph*{Algorithms for Continuous HuT in $d \leq 2$}
\label{sec:HuT-baseline-algorithms}
The algorithm of \cite{chew1999GeometricPatternMatching} mainly builds upon the idea that for $\mathcal Q = \{q + [-\delta, \delta]^d \mid q \in Q \}$ the region of feasible translations is exactly $\bigcap_{p \in P} \bigcup_{\hat Q \in \mathcal Q} \hat Q - p$.
For each point $p \in P$, they decompose $\mathcal Q - p$ into $O(|\mathcal Q|)$ disjoint boxes, which is possible in $d \leq 2$.
Given these $nm$ boxes, we compute the maximum depth among these.
Note that a point of depth $n$, which is the maximum depth possible, is a feasible translation to Hausdorff distance under translation.
Computing the maximum depth is possible using a segment tree in time $O(nm \log nm)$. 

For the undirected variant, we compute all regions of depth $n$. Afterwards, we switch the roles of $P, Q$ and use the same approach again, but we keep the resulting regions of maximum depth of the first iteration and only look for a region of depth $n+1$.

\end{document}